\documentclass[11pt,english]{article}
\usepackage{color}
\usepackage[T1]{fontenc}
\usepackage[latin9]{inputenc}
\usepackage[a4paper]{geometry}
\usepackage{textcomp}
\usepackage{amsthm}
\usepackage{amsmath}
\usepackage{amssymb}
\usepackage{esint}
\usepackage{bbm}
\usepackage{hyperref}
\usepackage{babel}

\usepackage{amscd}

\usepackage{amsfonts}
\usepackage{amsthm}
\usepackage{mathrsfs}
\usepackage{dsfont}

\usepackage{mathtools}

\usepackage{enumerate}
\usepackage{bm}
\usepackage{bbm}
\usepackage{verbatim}

\usepackage{enumitem}
\newlist{abbrv}{itemize}{1}
\setlist[abbrv,1]{label=,labelwidth=1.2in,align=parleft,itemsep=0.1\baselineskip,leftmargin=!}

\makeatletter

\DeclareFontEncoding{LGR}{}{}
\DeclareTextSymbol{\~}{LGR}{126}

\newcommand{\vph}{\widehat{\Phi}_{\kappa}}
\newcommand{\vphc}{\Phi_{\kappa}}
\newcommand{\vphp}{\widehat{\Phi}_{\kappa}^+}
\newcommand{\vphm}{\widehat{\Phi}_{\kappa}^-}
\newcommand{\vphpc}{\Phi_{\kappa}^+}
\newcommand{\vphmc}{\Phi_{\kappa}^-}

\newcommand{\tr}{\text{Tr}}

\newcommand{\Fp}{\hat{F}_{\kappa}^{+}}

\newcommand{\Fm}{\hat{F}_{\kappa}^{-}}

\newcommand{\Fpc}{F_{\kappa}^{+}}
\newcommand{\Fmc}{F_{\kappa}^{-}}

\newcommand{\abs}[1]{| #1 |}

\newcommand{\scp}[2]{\big\langle #1 , #2 \big\rangle}

\newcommand{\SCP}[2]{\big\langle #1 , #2 \big\rangle}

\newcommand{\bra}[1]{\langle #1 |}

\newcommand{\ket}[1]{| #1 \rangle}

\newcommand{\norm}[1]{\left|\left| #1 \right|\right|}

\renewcommand{\Re}{\mathrm{Re}}
\renewcommand{\Im}{\mathrm{Im}}

\newcommand{\id}{\mathbbm{1}}

\newcommand{\op}{\mathrm{op}}

\newcommand{\be}{\begin{equation}}
\newcommand{\ee}{\end{equation}}

\newtheorem{theorem}{Theorem}[section]
\newtheorem{lemma}[theorem]{Lemma}

\newtheorem{remark}[theorem]  {Remark}
\newtheorem{definition}[theorem] {Definition}

\newtheorem{conj}[theorem] {Conjecture}

\allowdisplaybreaks[1]

\begin{document}

\title{Mean-field limits of particles in interaction with quantized radiation fields}

\author{
Nikolai Leopold\footnote{
IST Austria (Institute of Science and Technology Austria), Am Campus 1, 3400 Klosterneuburg, Austria. E-mail: {\tt nikolai.leopold@ist.ac.at}} \ and 
Peter Pickl\footnote{Duke Kunshan University, Duke Avenue 8, 215316 Kunshan, China.
\newline
 E-mail: {\tt peter.pickl@dukekunshan.edu.cn}} \footnote{Ludwig-Maximilians-Universit\"at M\"unchen, Theresienstra\ss e 39, {80333} M\"unchen, Germany.
 \newline E-mail: {\tt pickl@math.lmu.de}}
}

\maketitle

\begin{center}
\textit{Dedicated to Herbert Spohn on the occasion of his 70th birthday.}
\end{center}

\begin{abstract}
\noindent
We report on a simple strategy to treat mean-field limits of quantum mechanical systems in which a large number of particles weakly couple to a second-quantized radiation field. Extending the method of counting, introduced in \cite{pickl1}, with ideas inspired by \cite{vytas} and \cite{falconi} leads to a technique that can be seen as a combination of the method of counting and the coherent state approach. It is similar to the coherent state approach but might be slightly better suited to systems in which a fixed number of particles couple to radiation. The strategy is effective and provides explicit error bounds.
As an instructional example we derive the Schr\"odinger-Klein-Gordon system of equations from the Nelson model with ultraviolet cutoff. Furthermore, we derive explicit bounds on the rate of convergence of the one-particle reduced density matrix of the non-relativistic particles in Sobolev norm.  More complicated models like the Pauli-Fierz Hamiltonian can be treated in a similar manner \cite{leopold}.
\end{abstract}

\noindent
\textbf{MSC class:} 35Q40, 81Q05, 82C10   

\section{Introduction}
\label{section: Nelson Introduction}

Quantum systems with many degrees of freedom are difficult to analyze. This is especially severe in the presence of quantized radiation fields which are described by Fock spaces with infinitely many degrees of freedom.
The dynamics of such systems can, however, be studied in special situations by means of simpler effective evolution equations. These involve fewer degrees of freedom, are less exact but easier to investigate. 
 Effective evolution equations for particles that interact with quantized radiation fields have rigorously been derived for example in \cite{ginibrenironivelo,falconi,ammarifalconi,teufel2,frankschlein,frankgang, leopold, griesemer}. 
The general setting in these works is given by the tensor product of two Hilbert spaces
\begin{align}
\mathcal{H}^{(N)} = \mathcal{H}_p^{(N)} \otimes \mathcal{F}.
\end{align} 
The space $\mathcal{H}_p^{(N)}$ describes $N$ non-relativistic particles and $\mathcal{F}$ (usually a bosonic Fock space) models  the quantized radiation field in terms of gauge bosons. 
The dynamics of the system is governed by the Schr\"odinger equation with a Hamiltonian of the form
\begin{align}
H^{N} \coloneqq H_0^{N} + H_{f} + \sum_{j=1}^N H_{int,j}.
\end{align}
Here, $H_0^{N}$ and $H_f$ (solely acting on $\mathcal{H}_p^{(N)}$ and $\mathcal{F}$) denote the free Hamiltonians of the particles and the radiation field. The term $H_{int,j}$ establishes an interaction between the $j$-th particle and the radiation field. This couples the dynamics of the particles with the gauge bosons.  
A typical question of interest is, whether the quantized radiation field can be approximated by a classical field and the evolution of the whole system described by a system of simple effective equations.
Usually one considers initial data $\Psi_{N,0} = \Phi_{N,0} \otimes W( \gamma^{1/2} \alpha_0) \Omega$ with no correlations between the particles and the gauge bosons, sometimes referred to as Pekar product state \cite{frankschlein}. The state $W( \gamma^{1/2} \alpha_0) \Omega \in \mathcal{F}$ denotes gauge bosons in the coherent state $\alpha_0$ with a mean particle number $\gamma \norm{\alpha_0}^2$, see \eqref{eq: Nelson Weyl operator}. Hereby, $\gamma$ is a model dependent scaling parameter, for instance the number of particles~\cite{falconi,ammarifalconi,leopold} or the strong coupling parameter in the Polaron model~\cite{frankschlein,frankgang,griesemer}. From physics literature it is commonly known that coherent states with a high occupation number of gauge bosons can approximately be described by a classical radiation field \cite[Chapter III.C.4]{cohen}. This allows us to describe the system in the limit $\gamma \rightarrow \infty$ (in a suitable sense, see Section~\ref{section: Nelson Main result}) effectively by the state of the particles $\Phi_{N,0}$ and a classical radiation field with mode function $\alpha_0$. The arising question is, if at later time $t$ one can still approximate the system by the pair $(\Phi_{N,t},\alpha_t)$ which evolves according to a set of simple effective equations with initial datum $(\Phi_{N,0},\alpha_0)$:
\begin{equation}
 \CD
  \Psi_{N,0} @>  \gamma \rightarrow \infty >>   (\Phi_{N,0},\alpha_0)
\\
      @V \text{Many-body dynamics} VV    @VV  \text{Effective dynamics}  V \\
\Psi_{N,t}  @>   \gamma \rightarrow \infty  >>  (\Phi_{N,t},\alpha_t) .
\endCD
\end{equation}
This only holds, if the radiation sector of $\Psi_{N,t}$ is approximately given by a coherent state, i.e. if the gauge bosons, that are created during the time evolution, are either in a coherent state or subleading  with respect to $\gamma$.
The effect of the particles on the radiation field is typically negligible, if one considers a fixed number of particles, a coupling constant that tends to zero in a suitable sense and a coherent state, whose mean particle number scales with the parameter $\gamma$ \cite{ginibrenironivelo}.
Otherwise, the state of the particles must have a special structure to ensure that the contributing gauge bosons are coherent \cite[Complement $\text{B}_{\text{III}}$]{cohen}. This is expected, if one considers slow and heavy particles~\cite{teufel2} or a condensate of particles that weakly couple to the radiation field. In this work, we are interested in the latter situation. More explicitly, we study the dynamics of initial states $\Psi_{N,0} = \varphi_0^{\otimes N} \otimes W(N^{1/2} \alpha_0) \Omega$ with one particle wave function $\varphi_0$ in the limit $N = \gamma \rightarrow \infty$ where the fields in the interaction Hamiltonian $H_{int,j}$ are multiplied by $N^{-1/2}$ (see Section~\ref{section: Nelson setting of the problem}). We refer to this limit as mean-field limit, because it implies that the source term of the radiation field is replaced by its mean value in the effective description. So far, such kind of limits have been studied either by the coherent state approach \cite{ginibrenironivelo,falconi2,falconi} or by means of Wigner measures \cite{ammarifalconi}.\footnote{These approaches usually embed the N particle states of $\mathcal{H}_p^{(N)}$ in a bosonic Fock space for the particles $\mathcal{F}_p$ and consider the Hilbert space $\mathcal{F}_p \otimes \mathcal{F}$.} While the method of Wigner measures allows us to derive limiting equations for an extensive class of initial states it does in contrast to the coherent state approach not provide quantitative bounds on the rate of convergence. 
In the following, we present a strategy, similar to the coherent state approach, which is designed for systems with fixed particle number. Such systems usually arise in the non-relativistic limit when the creation and annihilation of the non-relativistic particles is suppressed.\footnote{For the sake of clarity, we want to stress that only the number of the non-relativistic particles is fixed while gauge bosons are created and destroyed during the time evolution.}
The method provides explicit bounds on the rate of convergence and can be seen as a combination of the method of counting and the coherent state approach. 
As an instructional example we derive the Schr\"odinger-Klein-Gordon system of equations from the Nelson model with ultraviolet cutoff.
Our strategy seems general and we hope it will be useful in the treatment of more complicated models.
It was already applied to derive the Maxwell-Schr\"odinger system of equations from the spinless Pauli-Fierz Hamiltonian \cite{leopold}.

\section{Setting of the problem}
\label{section: Nelson setting of the problem}
We consider a system of N identical charged bosons interacting with a scalar field, described by a wave function $\Psi_{N,t} \in \mathcal{H}^{(N)}$. The Hilbert space is given by
\begin{align}
\mathcal{H}^{(N)} \coloneqq L^2\left( \mathbb{R}^{3N} \right) \otimes \mathcal{F} ,
\end{align}
where the scalar field is represented by elements of the Fock space 
$
 \mathcal{F}  \coloneqq  \bigoplus_{n \geq 0} L^2(\mathbb{R}^3)^{\otimes_s^n} .
$
 The subscript $s$ indicates symmetry under interchange of variables.  An element $\Psi_N \in \mathcal{H}^{(N)}$ is a sequence $\{ \Psi_N^{(n)} \}_{n \in \mathbb{N}_0}$ in $L^2(\mathbb{R}^{3N + 3n})$ 
 with\footnote{Note that $\Psi_N^{(n)}$ is symmetric in the variables $k_1, \ldots k_n$. For notational convenience we will use the shorthand notation $\Psi_N^{(n)}(X_N,K_n) = \Psi_N^{(n)}(x_1, \ldots, x_N, k_1, \ldots k_n)$.}
 \begin{align}
 \norm{\Psi_N}^2 = \sum_{n=0}^{\infty} \int d^{3N}x \, d^{3n}k \, \abs{\Psi_{N}^{(n)}(x_1,\ldots,x_N,k_1, \ldots, k_n)}^2
 < \infty.
 \end{align}
On $\mathcal{H}^{(N)}$, we define the (pointwise) annihilation and creation operators 
by\footnote{Here, $\hat{k}_j$ means that $k_j$ is left out in the argument of the function.}
\begin{align}
\label{eq: Nelson pointwise creation and annihilation operators}
\left( a(k) \Psi_N \right)^{(n)} (X_N,k_1,\ldots,k_n) &= ( n + 1)^{1/2} \Psi_N^{(n+1)}(X_N,k,k_1, \ldots, k_n) ,
\nonumber \\
\left( a^*(k) \Psi_N \right)^{(n)} (X_N,k_1,\ldots,k_n) &=
n^{-1/2} \sum_{j=1}^n \delta(k- k_j) \Psi_N^{(n)}(X_N, k_1, \ldots, \hat{k}_j, \ldots, k_n).
\end{align}
They are operator valued distributions and satisfy the commutation relations
\begin{align}
\label{eq: Nelson canonical commutation relation}
[a(k), a^*(l) ] &=  \delta(k-l), \quad
[a(k), a(l) ] = 
[a^*(k), a^*(l) ] = 0.
\end{align}
The time evolution of the wave function $\Psi_{N,t}$ is governed by the Schr\"odinger equation
 \begin{align}
\label{eq: Nelson Schroedinger equation microscopic}
 i \partial_t \Psi_{N,t} = H_N \Psi_{N,t}.
 \end{align}
 Here,
\begin{align}
\label{eq: Nelson Hamiltonian}
H_N =& \sum_{j=1}^N  \left( - \Delta_j  + \frac{\vph(x_j)}{\sqrt{N}} \right) + H_f
\end{align}
denotes the Nelson Hamiltonian and 
\begin{align}
\vph(x) =  \int d^3k \, \frac{\tilde{\kappa}(k)}{\sqrt{2 \omega(k)}} 
\left( e^{ikx} a(k) + e^{-ikx} a^*(k)  \right) .
\end{align}
The scalar bosons evolve according to the dispersion relation $\omega(k) = ( \abs{k}^2 + m_b^2 )^{1/2}$ with mass $m_b \geq 0$ and
\begin{equation}
\label{eq: Nelson cut off function}
\tilde{\kappa}(k) = (2 \pi)^{- 3/2} \ \id_{\abs{k}\leq \Lambda}(k), \quad
\text{with} \;  \id_{\abs{k}\leq \Lambda}(k) =
\begin{cases} 
1 &\text{if } \abs{k} \leq \Lambda , \\
0 &\text{otherwise}, 
\end{cases}
\end{equation}
cuts off the high frequency modes of the radiation field.
On the domain
\begin{align}
\label{eq: Nelson field energie domain}
\mathcal{D}(H_f) = \Big\{ \Psi_N \in \mathcal{H}^{(N)}: \sum_{n=1}^{\infty} \int d^{3N}x \, d^{3n}k \, \abs{\sum_{j=1}^n w(k_j)}^2  \abs{\Psi_N^{(n)}(X_N, K_n)}^2  < \infty  \Big\}
\end{align}
the free Hamiltonian of the scalar field is defined by
\begin{align}
\left( H_f \Psi_N \right)^{(n)} = \sum_{j=1}^n w(k_j) \Psi_N^{(n)}.
\end{align}
By means of the creation and annihilation operators it can be written as
\begin{align}
H_f &= \int d^3k \, \omega(k) a^*(k) a(k).
\end{align}
The Nelson model was originally introduced to describe the interaction of non-relativistic nucleons with a meson field. 
By standard estimates of the field operator and Kato's theorem it is easily shown that $H_N$ is a self-adjoint operator with $\mathcal{D}\left(H_N \right) = \mathcal{D} \big( \sum_{j=1}^N - \Delta_j +H_f \big)$ \cite{nelson, spohn}.
The scaling in front of the interaction ensures that the kinetic and potential energy of  $H_N$ are of the same order. 
For simplicity, we are first interested in the evolution of initial states of the product form
\begin{align}
\label{eq: Nelson initial product state}
\varphi_0^{\otimes N} \otimes W(\sqrt{N} \alpha_0) \Omega.
\end{align}
Here, $\Omega$ denotes the vacuum in $\mathcal{F}$ and $W(f)$ is the Weyl operator
\begin{align}
\label{eq: Nelson Weyl operator}
W(f) \coloneqq \exp \left(  \int d^3k \, f(k) a^*(k) - f^*(k) a(k) \right),
\end{align} 
where $f \in L^2(\mathbb{R}^3)$.
This choice of initial states corresponds to situations in which no correlations among the particles and the gauge bosons are present. 
Due to the interaction between the particles and the gauge bosons correlations take place and the time evolved state will no longer be of product form. However, for large $N$ and times of order one it can be approximated, in a sense more specified below, by a state of the form $\varphi_t^{\otimes N} \otimes W(\sqrt{N} \alpha_t) \Omega$, 
where $(\varphi_t,\alpha_t)$ solves the Schr\"odinger-Klein-Gordon equations\footnote{We use the shorthand notation $\left( \kappa * \Phi \right)(x,t) = \int d^3k \, e^{ikx} \tilde{\kappa}(k) \mathcal{FT}[\Phi](k,t)$, where $\mathcal{FT}[\Phi](k,t)$ denotes the Fourier transform of $\Phi(x,t)$.}
\begin{align}
\label{eq: Nelson Schroedinger-Klein-Gordon system}
\begin{cases}
i \partial_t \varphi_t(x) &= H^{eff} \varphi_t(x)
= \left[ - \Delta + \left( \kappa * \Phi \right)(x,t)  \right] \varphi_t(x),  \\
i \partial_t \alpha_t(k) &= \omega(k) \alpha_t(k) + (2 \pi)^{3/2} \frac{\tilde{\kappa}(k)}{\sqrt{2 \omega(k)}} \mathcal{FT}\left[ \abs{\varphi_t}^2 \right](k),  \\
 \Phi(x,t) &=
 \int d^3k \, (2 \pi)^{-3/2}   \frac{1}{\sqrt{2 \omega(k)}} 
\left(  e^{ikx} \alpha_t(k)  + e^{-ikx} \alpha^*_t(k)  \right),
\end{cases}
\end{align}
with $(\varphi_0,\alpha_0) \in L^2(\mathbb{R}^3) \oplus L^2(\mathbb{R}^3)$.
This system of equations determines the evolution of a single quantum particle in interaction with a classical scalar field. In the literature it is better known in its  formally equivalent form 
\begin{align}
\label{eq: Nelson Schroedinger-Klein-Gordon system 2}
\begin{cases}
i \partial_t \varphi_t(x) 
&=  \left[ - \Delta + \left( \kappa * \Phi \right)(x,t)  \right] \varphi_t(x),  \\
\left[ \partial_t^2 - \Delta + m_b^2 \right] \Phi(x,t) 
&= - \left( \kappa * \abs{\varphi_t}^2 \right)(x).
\end{cases}
\end{align}

\section{Main result}
\label{section: Nelson Main result}
The physical situation we are interested in is the dynamical description of a Bose-Einstein condensate of charges. We start initially with a product state~\eqref{eq: Nelson initial product state} and show that the condensate persists during the time evolution, i.e. correlations are small also at later times.
Let
\begin{align}
\label{eq: Nelson number operator definition}
\mathcal{N} \coloneqq \int d^3k \, a^*(k) a(k)
\end{align}
 be the number (of gauge bosons) operator
with domain
\begin{align}
\label{eq: Nelson number operator domain}
\mathcal{D}(\mathcal{N}) = \Big\{ \Psi_N \in \mathcal{H}^{(N)} : \sum_{n=1}^{\infty}  n^2 \int d^{3N}x \, d^{3n}k \,  \abs{\Psi_N^{(n)}(X_N, K_n)}^2  < \infty   \Big\}
\end{align}
and
$\Psi_{N,t} \in \left(L^2_s \left( \mathbb{R}^{3N} \right) \otimes \mathcal{F} \right) \cap  \mathcal{H}^{(N)} \cap \mathcal{D}(\mathcal{N})$ such that  $\norm{\Psi_{N,t}}_{\mathcal{H}^{(N)}}=1$. On the Hilbert space $L^2(\mathbb{R}^3)$ we define the "one-particle reduced density matrix of the charges" by
\begin{align}
\label{eq: Nelson definition reduced one-particle matrix charged particles}
\gamma_{N,t}^{(1,0)} \coloneqq \tr_{2,\ldots, N} \otimes \tr_{\mathcal{F}} \ket{\Psi_{N,t}} \bra{\Psi_{N,t}},
\end{align}
where $\tr_{2,\ldots, N}$  denotes the partial trace over the coordinates $x_2,\ldots, x_N$ and $\tr_{\mathcal{F}}$ the trace over Fock space. Then, the charged particles of the many-body state $\Psi_{N,t}$ are said to exhibit complete asymptotic Bose-Einstein condensation at time $t$, if there exists 
$\varphi_t \in L^2(\mathbb{R}^3)$ with $\norm{\varphi_t}=1$, such that
\begin{align}
\label{eq: Nelson convergence reduced one-particle matrix charged particles}
\tr_{L^2(\mathbb{R}^3)} \abs{\gamma_{N,t}^{(1,0)} - \ket{\varphi_t} \bra{\varphi_t}}  \rightarrow 0,
\end{align}
as $N \rightarrow \infty$. Such $\varphi_t$ is called the condensate wave function. For other indicators of condensation and their relation we refer to \cite{michelangeli}. Moreover, we introduce the "one-particle reduced density matrix of the gauge bosons" with kernel
\begin{align}
\label{eq: Nelson definition reduced one-particle matrix photon}
\gamma_{N,t}^{(0,1)}(k,k') \coloneqq N^{-1}  \scp{\Psi_{N,t}}{ a^*(k')  a(k) \Psi_{N,t}}_{\mathcal{H}^{(N)}}  .
\end{align}
$\gamma_{N,t}^{(0,1)}$ is a positive trace class operator with $\tr_{L^2(\mathbb{R}^3)}(\gamma_{N,t}^{(0,1)})= N^{-1} \scp{\Psi_{N,t}}{\mathcal{N} \Psi_{N,t}}_{\mathcal{H}^{(N)}}$.
It should be noted, that \eqref{eq: Nelson definition reduced one-particle matrix photon} differs from the usual definition (e.g. \cite[p.8]{rodnianskischlein}) by the weight factor  $ \scp{\Psi_N}{\mathcal{N} \Psi_N}_{\mathcal{H}^{(N)}}/N$. Our choice ensures that we only measure deviations from the classical mode function that are at least of order $N$. This is reasonable because Fock space vectors with a mean particle number smaller than of order $N$ only have a subleading effect on the dynamics of the charged particles. 
We say the gauge bosons exhibit "asymptotic Bose-Einstein condensation", if there exists a state $\alpha_t \in L^2(\mathbb{R}^3)$, such that
\begin{align}
\label{eq: Nelson convergence reduced one-particle matrix photon}
\tr_{L^2(\mathbb{R}^3)} \abs{\gamma_{N,t}^{(0,1)} - \ket{\alpha_t} \bra{\alpha_t}} \rightarrow 0,
\end{align}
 as $N \rightarrow \infty$. \\
In order to derive our main result, the solutions of the Schr\"odinger-Klein-Gordon equations have to satisfy the following assumptions.

\begin{definition}
\label{definition: assumptions on the solutions of the effective system}
Let $m \in \mathbb{N}$, $H^m(\mathbb{R}^3)$ denote the Sobolev space of order $m$ and $L_m^2(\mathbb{R}^3)$ a weighted $L^2$-space with norm
$\norm{\alpha}_{L^2_m(\mathbb{R}^3)} =  \big| \big|  ( 1 + \abs{\cdot}^2 )^{m/2} \alpha \big| \big|_{L^2(\mathbb{R}^3)}.$ We define two sets of solutions of the Schr\"odinger-Klein-Gordon equations:
\begin{align}
(\varphi_t,\alpha_t) \in \mathcal{G}_1 \Leftrightarrow  &(a) \quad (\varphi_t,\alpha_t)  \; \text{is a $L^2 \oplus L^2$ solution of \eqref{eq: Nelson Schroedinger-Klein-Gordon system} with $\norm{\varphi_t}_{L^2(\mathbb{R}^3)} =1$} 
\nonumber \\
  &(b) \quad (\varphi_t,\alpha_t) \in H^2(\mathbb{R}^3) \oplus L_1^2(\mathbb{R}^3).
\\
(\varphi_t,\alpha_t) \in \mathcal{G}_2 \Leftrightarrow  &(a) \quad (\varphi_t,\alpha_t)  \; \text{is a $L^2 \oplus L^2$ solution of \eqref{eq: Nelson Schroedinger-Klein-Gordon system} with $\norm{\varphi_t}_{L^2(\mathbb{R}^3)} =1$ } 
\nonumber \\
  &(b) \quad (\varphi_t,\alpha_t) \in H^4(\mathbb{R}^3) \oplus L_2^2(\mathbb{R}^3).
\end{align}
\end{definition}
\noindent
These assumptions are expected to follow from appropriately chosen initial data.\footnote{We suppose that Conjecture~\ref{lemma: Nelson global existence of Schroedinger Klein Gordon} can be proven by a standard fixed-point argument. Especially due to the cutoff in the radiation field it seems possible to make use of Theorem X.74 in~\cite{reedsimon}.}
\begin{conj}
\label{lemma: Nelson global existence of Schroedinger Klein Gordon}
Let $(\varphi_0, \alpha_0) \in  H^{2n}(\mathbb{R}^3) \oplus L_{n}^2(\mathbb{R}^3)$ for $1\leq n \leq2$. Then, there is a strongly differentiable $\left( H^{2n}(\mathbb{R}^3) \oplus L_{n}^2(\mathbb{R}^3) \right)$-valued function $(\varphi(t),\alpha(t))$ on $[0, \infty )$ that satisfies \eqref{eq: Nelson Schroedinger-Klein-Gordon system}. 
\end{conj}

\noindent
Our main theorem is the following.

\begin{theorem}
\label{theorem: Nelson main theorem}
Let $(\varphi_t, \alpha_t) \in \mathcal{G}_1$ and $ \Psi_{N,0} \in \left( L_s^2(\mathbb{R}^{3N})  \otimes \mathcal{F} \right) \cap \mathcal{D} \left( \mathcal{N} \right) \cap \mathcal{D} \left(  \mathcal{N}  H_N \right)$ with $\norm{\Psi_{N,0}} =1$
such that \footnote{Here, $W^{-1}(\sqrt{N} \alpha_0) = W(- \sqrt{N} \alpha_0)$ is the inverse of the unitary Weyl operator $W(\sqrt{N} \alpha_0)$, see Section~\ref{section Nelson initial states}.} 
\begin{align}
a_N =& Tr_{L^2(\mathbb{R}^3)} \abs{\gamma_{N,0}^{(1,0)} - \ket{\varphi_0} \bra{\varphi_0}} \rightarrow 0 \; and  \\
b_N =& N^{-1} \scp{W^{-1}(\sqrt{N} \alpha_0) \Psi_{N,0}}{ \mathcal{N} W^{-1}(\sqrt{N} \alpha_0) \Psi_{N,0}}_{\mathcal{H}^{(N)}} \rightarrow 0
\end{align}
as $N \rightarrow \infty$.
Let $\Psi_{N,t}$ be the unique solution of \eqref{eq: Nelson Schroedinger equation microscopic} with initial data $\Psi_{N,0}$. Then, there exists a generic constant $C$ independent of $N$, $\Lambda$ and $t$ such that
\begin{align}
 \label{eq: Nelson main theorem 1}
\text{Tr}_{L^2(\mathbb{R}^3)} \abs{\gamma_{N,t}^{(1,0)} - \ket{\varphi_t} \bra{\varphi_t}} &\leq
\sqrt{a_N + b_N + N^{-1}} e^{ \Lambda^2 C t} , \\
\label{eq: Nelson main theorem 2}
\text{Tr}_{L^2(\mathbb{R}^3)} \abs{\gamma_{N,t}^{(0,1)} - \ket{\alpha_t} \bra{\alpha_t}} &\leq
\sqrt{a_N + b_N + N^{-1}}  e^{\Lambda^2 C t} C \left( 1 + \norm{\alpha_t} \right)
\end{align}
for any $t \in \mathbb{R}^+_0$.\footnote{To ease the presentation we have chosen for given $t$ the scaling parameter $N$ large enough such that $0 \leq \beta(t) \leq 1$ and $0 \leq \beta_2(t) \leq 1$ (see Subsections~\ref{subsec: Estimate on the time derivative of beta} and \ref{subsec:: Control of the kinetic energy}).}
In particular, for $\Psi_{N,0} = \varphi_0^{\otimes N} \otimes W(\sqrt{N} \alpha_0) \Omega$ one obtains
\begin{align}
 \label{eq: Nelson main theorem 4}
\text{Tr}_{L^2(\mathbb{R}^3)} \abs{\gamma_{N,t}^{(1,0)} - \ket{\varphi_t} \bra{\varphi_t}} &\leq
N^{-1/2}   e^{C \Lambda^2 t} , \\
\label{eq: Nelson main theorem 5}
\text{Tr}_{L^2(\mathbb{R}^3)} \abs{\gamma_{N,t}^{(0,1)} - \ket{\alpha_t} \bra{\alpha_t}} &\leq
N^{-1/2} e^{\Lambda^2 C t} C \left( 1 + \norm{\alpha_t} \right) .
\end{align}
Moreover, let $(\varphi_t, \alpha_t) \in \mathcal{G}_2$ and $ \Psi_{N,0} \in \left( L_s^2(\mathbb{R}^{3N})  \otimes \mathcal{F} \right) \cap \mathcal{D} \left( \mathcal{N} \right) \cap \mathcal{D} \left(  \mathcal{N}  H_N \right) \cap \mathcal{D} \left( H_N^2 \right)$ such that
\begin{align}
c_N =& \norm{\nabla_1 \left( 1 - \ket{\varphi_0} \bra{\varphi_0} \otimes \id_{L^2 (\mathbb{R}^{3(N-1)} )} \otimes \id_{\mathcal{F}} \right) \Psi_{N,0}}^2_{\mathcal{H}^{(N)}} \rightarrow 0 
\end{align}
as $N \rightarrow \infty$. Then, there exists a positive monotone increasing function $C(s)$ of the norms $\norm{\alpha_s}_{L^2(\mathbb{R}^3)}$ and $\norm{\varphi_s}_{H^1(\mathbb{R}^3)}$ such that
\begin{align}
 \label{eq: Nelson main theorem 3}
\text{Tr}_{L^2(\mathbb{R}^3)} \abs{\sqrt{1 - \Delta} \left(\gamma_{N,t}^{(1,0)} - \ket{\varphi_t} \bra{\varphi_t}\right) \sqrt{1 - \Delta}} &\leq
\sqrt{a_N + b_N + c_N + N^{-1} } C(t)  e^{ \Lambda^4 \int_0^t C(s) ds} .
\end{align}
For $\Psi_{N,0} = \varphi_0^{\otimes N} \otimes W(\sqrt{N} \alpha_0) \Omega$ one obtains
\begin{align}
\label{eq: Nelson main theorem 6}
\text{Tr}_{L^2(\mathbb{R}^3)} \abs{\sqrt{1 - \Delta} \left(\gamma_{N,t}^{(1,0)} - \ket{\varphi_t} \bra{\varphi_t}\right) \sqrt{1 - \Delta}} &\leq N^{-1/2} C(t) e^{ \Lambda^4 \int_0^t C(s) ds}.
\end{align}
\end{theorem}

\begin{remark}
\label{remark: Nelson remark to the main theorem}
The convergence of the reduced density matrices in trace norm with rate $N^{-1}$ was already shown in~\cite{falconi} for special classes of initial states (coherent and product states).\footnote{For a precise definition we refer to~\cite[Theorem 3]{falconi}.}
Theorem~\ref{theorem: Nelson main theorem} extends this result to more general states but only with error estimates of order $N^{-1/2}$. Moreover, we present the first explicit bounds on the rate of convergence of the one-particle reduced density matrix of the charges in Sobolev norm. It seems possible to obtain the convergence rate $N^{-1}$, if one regards (similar to~\cite{picklnorm}) fluctuations around the mean-field dynamics. 
\end{remark}

\section{Comparison with the literature}

In \cite{ginibrenironivelo}, Ginibre, Nironi and Velo derived the Schr\"odinger-Klein-Gordon system of equations from the Nelson model with cutoff.  They considered a finite number of charged bosons, a coupling constant that tends to zero and a coherent state of gauge bosons whose particle number goes to infinity. The number of gauge bosons that are created during the time evolution is negligible in this case and it is possible to approximate the quantized scalar field by an external potential which evolves according to the Klein-Gordon equation without source term. Falconi~\cite{falconi} derived the Schr\"odinger-Klein-Gordon system of equations in the setting of the present paper by means of the coherent state approach. A comparison between his result and Theorem~\ref{theorem: Nelson main theorem} is given in Remark~\ref{remark: Nelson remark to the main theorem}.
Making use of a Wigner measure approach Ammari and Falconi \cite{ammarifalconi} were able to establish the classical limit (without quantitative bounds on the rate of convergence) of the renormalized Nelson model without cutoff.
Teufel \cite{teufel2} considered the adiabatic limit of the Nelson model and showed that the interaction mediated by the quantized radiation field is well approximated by a direct Coulomb interaction.
In \cite{leopold} we used the strategy of the present paper to derive the Maxwell-Schr\"odinger equations from the spinless Pauli-Fierz Hamiltonian. Here, additional technical difficulties arise from the minimal coupling term in the Pauli-Fierz Hamiltonian.

\section{Notations}
The Fourier transform of a function $f$ is denoted by $\tilde{f}$ or $\mathcal{FT}[f]$.
$H^s(\mathbb{R}^3)$  stands for the Sobolev space with norm 
$\norm{f}_{H^s(\mathbb{R}^3)} =  \big| \big|  ( 1 + \abs{\cdot}^2 )^{s/2} \mathcal{FT}[f] \big| \big|_{L^2(\mathbb{R}^3)}$ and $L^2_m(\mathbb{R}^3)$ is the weighted $L^2$ space with  
$\norm{f}_{L^2_m(\mathbb{R}^3)} =  \big| \big|  ( 1 + \abs{\cdot}^2 )^{m/2} f \big| \big|_{L^2(\mathbb{R}^3)}$. 
Moreover, we use $\norm{A}_{HS} = \sqrt{Tr A^* A}$ to denote the Hilbert-Schmidt norm.
With a slight abuse of notation we write $\Phi$ and $F$ to indicate the scalar and auxiliary field but also their respective Fourier transforms. If we use $\Phi(t)$ or $F(t)$,  we always refer to the coordinate representation of the fields.
Furthermore, we apply the shorthand notation $\vphc(x,t) \coloneqq \left( \kappa * \Phi \right)(x,t)$.  \\

\section{The strategy}
We are interested in the evolution of product states of the form  \eqref{eq: Nelson initial product state} under the dynamics~\eqref{eq: Nelson Schroedinger equation microscopic}. The scalar field in the Nelson Hamiltonian establishes an interaction between the charges and the field modes with wave vectors smaller than $\Lambda$.\footnote{One should note that the high frequency modes of the radiation field do not interact with the non-relativistic particles and evolve according to the free dynamics.}
This changes the state of the charges, leads to the creation and annihilation of gauge bosons and causes initially factorized states to build correlations between the charges, the gauge bosons as well as among charges and gauge bosons.
To study the emergence of these correlations we combine the  "method of counting", introduced in \cite{pickl1}, with ideas from \cite{vytas} and \cite{falconi}. The result can be seen as a fusion of the "method of counting" and the coherent state approach, as used for instance in \cite{falconi, rodnianskischlein}. The key idea is to prove condensation not in terms of reduced density matrices but to consider a different indicator of condensation. To study the correlations between the charges we introduce a functional $\beta^a$, which counts the relative number of particles that are not in the state of the condensate wave function $\varphi_t$.
\begin{definition}
For any $N \in \mathbb{N}$, $\varphi_t \in L^2(\mathbb{R}^3)$ with $\norm{\varphi_t}=1$ and $1 \leq j \leq N$ we define the time-dependent projectors
$p_j^{\varphi_t}: L^2(\mathbb{R}^{3N}) \rightarrow  L^2(\mathbb{R}^{3N})$ and
$q_j^{\varphi_t}: L^2(\mathbb{R}^{3N}) \rightarrow  L^2(\mathbb{R}^{3N})$  by
\begin{align}
p_j^{\varphi_t} f(x_1, \ldots, x_N)
\coloneqq \varphi_t(x_j) \int d^3x_j \, \varphi_t^*(x_j) f(x_1, \ldots, x_N) 
\quad \text{for all} \; f \in L^2(\mathbb{R}^{3N})
\end{align}
and $q_j^{\varphi_t} \coloneqq 1 - p_j^{\varphi_t}$.\footnote{Ocassionally, we use the bra-ket notation $p_j^{\varphi_t} = \ket{\varphi_t(x_j)} \bra{\varphi_t(x_j)} = \ket{\varphi_t} \bra{\varphi_t}_j$.}
Let $\Psi_{N,t} \in \mathcal{H}^{(N)}$. Then $\beta^a: \mathcal{H}^{(N)} \times L^2(\mathbb{R}^3) \rightarrow \mathbb{R}_0^+$ is given by
\begin{align}
\beta^a\left( \Psi_{N,t},\varphi_t \right) \coloneqq \SCP{\Psi_{N,t}}{q_1^{\varphi_t} \otimes \id_{\mathcal{F}} \, \Psi_{N,t}}.
\end{align}
\end{definition}

\begin{remark} 
The functional $\beta^a$ was already used in
\cite{pickl1,pickl2,picklgp3d,pickl4,knowles, michelangeli2,  picklnorm, anapolitanosmott}
and others to derive the Hartree and Gross-Pitaevskii equation.
\end{remark}

\noindent
The situation is slightly different in the radiation sector
 because the number of gauge bosons is not preserved during the time evolution. Moreover, it is known from physics literature~\cite[Chapter III.C.4]{cohen} that
the radiation field must be in a coherent state with a high occupation number of gauge bosons to behave classically. This is a state not only with little correlations but also a Poisson distributed number of gauge bosons. In order to investigate if the state of the radiation field is coherent we define a functional, referred to as $\beta^b$, which measures the fluctuations of the field modes around the classical mode function $\alpha_t$ for each time.
\begin{definition}
Let $\alpha_t \in L^2(\mathbb{R}^3)$ and $\Psi_{N,t} \in \mathcal{H}^{(N)} \cap \mathcal{D}\left(\mathcal{N} \right)$. Then $\beta^b: \mathcal{H}^{(N)} \cap \mathcal{D}\left(\mathcal{N} \right) \times L^2(\mathbb{R}^3) \rightarrow \mathbb{R}_0^+$ is given by
\begin{align}
\beta^b\left( \Psi_{N,t}, \alpha_t \right) 
\coloneqq \int d^3k \, \SCP{\left( \frac{a(k)}{\sqrt{N}} - \alpha_t(k) \right) \Psi_{N,t}}{ \left( \frac{a(k)}{\sqrt{N}} - \alpha_t(k) \right) \Psi_{N,t}}.
\end{align} 
\end{definition}

\begin{remark}
Let $\alpha_0 \in L^2(\mathbb{R}^3)$ and $\Psi_{N,0}= W(\sqrt{N} \alpha_0) \Psi$ for some $\Psi\in \mathcal{H}^{(N)} \cap\mathcal{D}\left(\mathcal{N} \right)$. Then, the functional $\beta^b$ can be written as
\begin{align}
\label{eq: Nelson fluctuation dynamics and beta-b}
\beta^b\left(\Psi_{N,t},\alpha_t \right) = N^{-1} \SCP{\mathcal{U}_N(t;0) \Psi}{\mathcal{N} \mathcal{U}_N(t;0) \Psi},
\end{align}
where $\mathcal{U}_N(t;0) = W^*(\sqrt{N} \alpha_t) e^{- i H_Nt} W(\sqrt{N} \alpha_0)$ denotes the fluctuation dynamics of the coherent state approach (as used for example in \cite[p.18]{schleinbook}).\footnote{This is a simple consequence of $W(\sqrt{N} \alpha_t)$ being unitary and $W^*(\sqrt{N} \alpha_t) a(k) = a(k) W^*(\sqrt{N} \alpha_t) + \sqrt{N} W^*(\sqrt{N} \alpha_t) \alpha_t(k)$, see \eqref{eq: Nelson relation between beta-b and fluctuation dynamics}.} Thus, $\beta^b$ measures the number of gauge boson fluctuations around the effective evolution.
\end{remark}
\begin{remark}
It seems that $\beta^a$ is the natural quantity to consider for condensates with fixed particle number. The functional $\beta^b$, which usually arises in the coherent state approach as used in~\cite{rodnianskischlein, falconi,schleinbook} and others, is perfectly suited to keep track if the state of the radiation field remains coherent.
\end{remark}
\noindent
 Finally, the counting functional is defined by
\begin{definition}
Let $N \in \mathbb{N}$, $\varphi_t \in L^2(\mathbb{R}^3)$ with $\norm{\varphi_t}=1$, $\alpha_t \in L^2(\mathbb{R}^3)$ and $\Psi_{N,t} \in \mathcal{H}^{(N)} \cap \mathcal{D}\left(\mathcal{N} \right)$. Then $\beta: \mathcal{H}^{(N)} \cap \mathcal{D}\left(\mathcal{N} \right) \times L^2(\mathbb{R}^3) \times L^2(\mathbb{R}^3) \rightarrow \mathbb{R}_0^+$ is 
defined by\footnote{We sometimes apply the shorthand notation $\beta(t) = \beta(\Psi_{N,t},\varphi_t,\alpha_t)$.}
\begin{align}
\beta \left( \Psi_{N,t}, \varphi_t, \alpha_t \right) 
\coloneqq  \beta^a\left( \Psi_{N,t},\varphi_t \right)
+  \beta^b\left( \Psi_{N,t}, \alpha_t \right).
\end{align}
\end{definition}

\noindent
In summary, the functional has the following properties:
\begin{itemize}
\item[(i)] $\beta^a$ measures if the non-relativistic particles exhibit condensation.
\item[(ii)] $\beta^b$ examines whether the radiation field is in a coherent state.
\item[(iii)] $\beta \left(\Psi_{N,t},\varphi_t,\alpha_t \right)  \rightarrow 0$ as $N \rightarrow \infty$ implies condensation in terms of reduced density matrices (Lemma~\ref{lemma: Nelson relation between beta and reduced density matrices}).
\item[(iv)] $\beta \left(\Psi_{N,t},\varphi_t,\alpha_t \right) = 0$ if $\Psi_{N,t} = \varphi_t^{\otimes N} \otimes W(\sqrt{N} \alpha_t) \Omega$ (see Lemma~\ref{lemma: Nelson coherent states as initial states}).

\end{itemize}
In order to show that the product structure \eqref{eq: Nelson initial product state} is preserved during the time evolution we apply the following strategy
\begin{enumerate}
\item We choose initial states $\varphi_0, \alpha_0$ and $\Psi_{N,0}$ such that $\beta\left( \Psi_{N,0}, \varphi_0, \alpha_0 \right) \leq a_N + b_N  \rightarrow 0$ as $N \rightarrow \infty$.
\item For each $t \in \mathbb{R}_0^+$ we estimate $\abs{d_t \beta \left( \Psi_{N,t}, \varphi_t, \alpha_t \right) } \leq C \Lambda^2 \left( \beta \left( \Psi_{N,t}, \varphi_t, \alpha_t \right)  + N^{-1} \right)$ for some $C \in \mathbb{R}_0^+$. Then, Gr\"onwall's Lemma establishes the bound
$\beta \left( \Psi_{N,t}, \varphi_t, \alpha_t \right) 
\leq e^{C \Lambda^2 t} \left( \beta\left( \Psi_{N,0}, \varphi_0, \alpha_0 \right) + N^{-1} \right)$.
\item By means of property (iii) we conclude  condensation in terms of reduced density matrices.
\end{enumerate}

\noindent
To show the convergence of $\gamma_{N,t}^{(1,0)}$ to the projector onto the condensate wave function in Sobolev norm we include $\beta^c(\Psi_{N,t},\varphi_t) \coloneqq \norm{\nabla_1 q^{\varphi_t}_1 \Psi_{N,t}}^2$ in the definition of the functional. This allows us to control the kinetic energy of the non-relativistic particles which are not in the condensate.
\begin{definition}
Let $N \in \mathbb{N}$, $\varphi_t \in H^2(\mathbb{R}^3)$ with $\norm{\varphi_t}=1$, $\alpha_t \in L^2(\mathbb{R}^3)$ and $\Psi_{N,t} \in \mathcal{D}(H_N) \cap \mathcal{D}\left(\mathcal{N} \right)$. Then $\beta_2: \mathcal{D}(H_N) \cap \mathcal{D}\left(\mathcal{N} \right) \times H^2(\mathbb{R}^3) \times L^2(\mathbb{R}^3) \rightarrow \mathbb{R}_0^+$ is defined by
\begin{align}
\beta_2 \left( \Psi_{N,t}, \varphi_t, \alpha_t \right) 
&\coloneqq  \beta \left( \Psi_{N,t}, \varphi_t, \alpha_t \right)  + \beta^c\left(\Psi_{N,t},\varphi_t \right)
\nonumber \\
&\, =
\beta \left( \Psi_{N,t}, \varphi_t, \alpha_t \right) 
+  \norm{\nabla_1 q_1^{\varphi_t} \Psi_{N,t}}^2.
\end{align}
\end{definition}
\noindent
We would like to remark, that the ultraviolet cutoff~\eqref{eq: Nelson cut off function} is essential for the proof because:
\begin{enumerate}
\item The finiteness of $\norm{\eta}_2$ (see \eqref{eq: Nelson cutoff functions norm}) is needed to establish a connection between the difference of the radiation fields and the functional $\beta^b$ by means of the auxiliary fields~\eqref{eq: Nelson auxilliary fields}.

\item  The cutoff $\Lambda$ imposes regularity on the radiation fields which will be used to estimate the time derivative of $\norm{\nabla_1 q_1 \Psi_{N,t}}^2$. In spirit, this is opposite to the usual treatment of the polaron~\cite{liebthomas}, where regularity of the electron state is used to obtain a sufficient decay in the field modes with large wave vectors.
\end{enumerate}

\section{Relation to reduced density matrices}
\label{sec: Nelson Relation to reduced density matrices}

In this section, we relate the functional $\beta$ to the  trace norm distance of the one-particle reduced density matrices.

\begin{lemma}
\label{lemma: Nelson relation between beta and reduced density matrices}
Let $N \in \mathbb{N}$, $\varphi_t \in L^2(\mathbb{R}^3)$ with $\norm{\varphi_t}=1$, $\alpha_t \in L^2(\mathbb{R}^3)$ and $\Psi_{N,t} \in \mathcal{H}^{(N)} \cap \mathcal{D}\left(\mathcal{N} \right)$. Then,
\begin{align}
\label{eq: Nelson relation between beta and reduced density matrices 1}
\beta^a(\Psi_{N,t},\varphi_t) &\leq  \text{Tr}_{L^2(\mathbb{R}^3)} \abs{\gamma_{N,t}^{(1,0)} - \ket{\varphi_t} \bra{\varphi_t}} \leq  \sqrt{8  \beta^a(\Psi_{N,t},\varphi_t)}  ,  \\
\label{eq: Nelson relation between beta and reduced density matrices 2}
\text{Tr}_{L^2(\mathbb{R}^3)} \abs{\gamma_{N,t}^{(0,1)} - \ket{\alpha_t}\bra{\alpha_t}}  &\leq  3 \beta^b(\Psi_{N,t},\alpha_t) + 6 \norm{\alpha_t}_{L^2(\mathbb{R}^3)} \sqrt{\beta^b(\Psi_{N,t},\alpha_t)}. 
\end{align}
For $\varphi_t \in H^2(\mathbb{R}^3)$ with $\norm{\varphi_t}=1$ and $\Psi_{N,t} \in \mathcal{H}^{(N)} \cap \mathcal{D}(H_N)$, we have
\begin{align}
\label{eq: Nelson relation between beta and reduced density matrices 4}
&\text{Tr}_{L^2(\mathbb{R}^3)} \abs{ \sqrt{1- \Delta} \left(\gamma_{N,t}^{(1,0)} - \ket{\varphi_t} \bra{\varphi_t} \right) \sqrt{1 - \Delta}}  \leq
\big(1 + \norm{\varphi_t}_{H^1(\mathbb{R}^3)}^2 \big)  \times   
 \nonumber\\
&\times \left( \beta^a(\Psi_{N,t},\varphi_t) + \beta^c(\Psi_{N,t},\varphi_t)   \right) + 2 \norm{\varphi_t}_{H^1(\mathbb{R}^3)} \sqrt{\beta^a(\Psi_{N,t},\varphi_t) + \beta^c(\Psi_{N,t},\varphi_t)}.
\end{align}
\end{lemma}
\begin{proof}
The lower bound of~\eqref{eq: Nelson relation between beta and reduced density matrices 1} is proven by
\begin{align}
\beta^a(t) &= 1 - \SCP{\Psi_{N,t}}{p_1^{\varphi_t} \Psi_{N,t}}  =
1 - \scp{\varphi_t}{\gamma_{N,t}^{(1,0)} \varphi_t}
= \text{Tr}_{L^2(\mathbb{R}^3)} ( \ket{\varphi_t} \bra{\varphi_t} - \ket{\varphi_t} \bra{\varphi_t} \gamma_{N,t}^{(1,0)})  
\nonumber \\
&\leq \norm{p_1}_{\op}  \text{Tr}_{L^2(\mathbb{R}^3)} \abs{\gamma_{N,t}^{(1,0)} - \ket{\varphi_t} \bra{\varphi_t}}
=  \text{Tr}_{L^2(\mathbb{R}^3)} \abs{\gamma_{N,t}^{(1,0)} - \ket{\varphi_t} \bra{\varphi_t}}.
\end{align}
To obtain the upper bound we use that
\begin{align}
\label{eq: Nelson bound mit Hilbert Schmidt norm}
\text{Tr} \abs{\gamma - p} \leq 2 \norm{\gamma - p}_{HS}
+ \text{Tr} (\gamma - p)
\end{align}
is valid for any one-dimensional projector $p$ and non-negative density matrix $\gamma$. The original argument of the proof  was first observed by Robert Seiringer, see \cite{rodnianskischlein}. We present a version that is found in \cite{anapolitanosmott}: 
Let $(\lambda_n)_{n \in \mathbb{N}}$ be the sequence of eigenvalues of the trace class operator $A:= \gamma - p$.
Since $p$ is a rank one projection, $A$ has at most one negative eigenvalue.
If there is no negative eigenvalue, $\text{Tr} \abs{A} = \text{Tr} (A)$ and \eqref{eq: Nelson bound mit Hilbert Schmidt norm} holds. If there is one negative eigenvalue $\lambda_1$, we have
$\text{Tr} \abs{A} = \abs{\lambda_1} + \sum_{n \geq 2} \lambda_n = 2 \abs{\lambda_1} + \text{Tr}(A). $
Inequality~\eqref{eq: Nelson bound mit Hilbert Schmidt norm} then follows from $\abs{\lambda_1} \leq \norm{A}_{\op} \leq \norm{A}_{HS}.$\\
This shows
\begin{align}
\text{Tr}_{L^2(\mathbb{R}^3)} \abs{\gamma_{N,t}^{(1,0)} - \ket{\varphi_t} \bra{\varphi_t}} \leq 2 \norm{\gamma_{N,t}^{(1,0)} - \ket{\varphi_t} \bra{\varphi_t}}_{HS}
\end{align}
because $\text{Tr}_{L^2(\mathbb{R}^3)} (\gamma_{N,t}^{(1,0)} - \ket{\varphi_t} \bra{\varphi_t}) = 0$. The upper bound of~\eqref{eq: Nelson relation between beta and reduced density matrices 1} is obtained by
\begin{align}
\text{Tr}_{L^2(\mathbb{R}^3)} ( \gamma_{N,t}^{(1,0)} - \ket{\varphi_t} \bra{\varphi_t} )^2
=& 1 - 2  \text{Tr}_{L^2(\mathbb{R}^3)}  (\ket{\varphi_t} \bra{\varphi_t} \gamma_{N,t}^{(1,0)} ) +  \text{Tr}_{L^2(\mathbb{R}^3)} ((\gamma_{N,t}^{(1,0)})^2)
\nonumber \\
\leq& 2 (1 - \text{Tr}_{L^2(\mathbb{R}^3)}  (\ket{\varphi_t} \bra{\varphi_t} \gamma_{N,t}^{(1,0)} ))
= 2 \beta^a(t).
\end{align}
To prove~\eqref{eq: Nelson relation between beta and reduced density matrices 2} it is useful to write the kernel of 
$\gamma_{N,t}^{(0,1)} - \ket{\alpha_t}\bra{\alpha_t}$ as
\begin{align}
(\gamma_{N,t}^{(0,1)} - \ket{\alpha_t}\bra{\alpha_t})(k,l)  
&=   N^{-1} \SCP{\Psi_N}{a^*(l) a(k) \Psi_N} - \alpha_t^*(l) \alpha_t(k) 
\nonumber \\
&=  \SCP{\left(  N^{-1/2} a(l) - \alpha_t(l) \right)\Psi_N}{\left(  N^{-1/2}  a(k) - \alpha_t(k) \right)\Psi_N} 
\nonumber \\
&+  \alpha_t(k)  \SCP{\left(  N^{-1/2} a(l) - \alpha_t(l) \right)\Psi_N}{\Psi_N}  
\nonumber\\
&+   \alpha_t^*(l) \SCP{\Psi_N}{\left(  N^{-1/2} a(k) - \alpha_t(k) \right)\Psi_N}.
\end{align}
By means of Schwarz's inequality we have
\begin{align}
\abs{(\gamma_{N,t}^{(0,1)} - \ket{\alpha_t}\bra{\alpha_t})(k,l)}^2 
&\leq \norm{ \left(  N^{-1/2}  a(k) - \alpha_t(k) \right) \Psi_{N}}^2 
\norm{\left(  N^{-1/2} a(l) - \alpha_t(l) \right) \Psi_{N}}^2 
\nonumber \\
&+ \abs{\alpha_t(l)}^2 \norm{\left(  N^{-1/2}  a(k) - \alpha_t(k) \right) \Psi_{N}}^2
\nonumber \\
&+  \abs{\alpha_t(k)}^2 \norm{\left(  N^{-1/2} a(l) - \alpha_t(l) \right) \Psi_{N}}^2 
\end{align}
and 
\begin{align}
\norm{\gamma_{N,t}^{(0,1)} - \ket{\alpha_t}\bra{\alpha_t}}_{HS}^2 &=  
\int d^3k \, \int d^3l \, \abs{(\gamma_{N,t}^{(0,1)} - \ket{\alpha_t}\bra{\alpha_t})(k,l)}^2
\nonumber \\
&\leq (\beta^b(t))^2 + 2 \norm{\alpha_t}_{L^2(\mathbb{R}^3)}^2 \beta^b(t).
\end{align}
Similarly, one obtains
\begin{align}
\text{Tr}_{L^2(\mathbb{R}^3)} (\gamma_{N,t}^{(0,1)} - \ket{\alpha_t}\bra{\alpha_t})
&\leq  \; \; \,  \int d^3k  \,
\abs{ (\gamma_{N,t}^{(0,1)} - \ket{\alpha_t}\bra{\alpha_t})(k,k)}
 \nonumber \\
&\leq \; \; \,   \int d^3k \,  \norm{\left( N^{-1/2} a(k) - \alpha_t(k) \right) \Psi_N}_{\mathcal{H}^{(N)}}^2
\nonumber \\
&+ 2  \int d^3k \, \abs{\alpha_t(k)}  \norm{\left( N^{-1/2} a(k) - \alpha_t(k) \right) \Psi_N}_{\mathcal{H}^{(N)}}.
\end{align}
Applying Schwarz's inequality in the second line leads to
\begin{align}
&\text{Tr}_{L^2(\mathbb{R}^3)} (\gamma_{N,t}^{(0,1)} - \ket{\alpha_t}\bra{\alpha_t}) 
\leq  \beta^b(t) + 2 \norm{\alpha_t}_{L^2(\mathbb{R}^3)} \sqrt{\beta^b(t)}.
\end{align}
Inequality~\eqref{eq: Nelson relation between beta and reduced density matrices 2} follows from the monotonicity of the square root and \eqref{eq: Nelson bound mit Hilbert Schmidt norm}.
The estimate~\eqref{eq: Nelson relation between beta and reduced density matrices 4} originates from~\cite{picklnorm}.  One starts with the relation
\begin{align}
Tr_{L^2(\mathbb{R}^3)} &\abs{\sqrt{1 - \Delta} ( \gamma_{N,t}^{(1,0)} - \ket{\varphi_t} \bra{\varphi_t} ) \sqrt{1 -\Delta}} 
\nonumber \\
&= 
\sup_{\norm{A_1} \leq 1} \abs{Tr_{L^2(\mathbb{R}^3)} ( A_1 \sqrt{1- \Delta}  (  \gamma_{N,t}^{(1,0)} - \ket{\varphi_t} \bra{\varphi_t} ) \sqrt{1 - \Delta} )},
\end{align}
where the supremum is applied to all compact operators $A_1$ on $L^2(\mathbb{R}^3)$ with norm smaller or equal to one. Then, one continues with
\begin{align}
&Tr_{L^2(\mathbb{R}^3)} (A_1 \sqrt{1 - \Delta_1} (\gamma_{N,t}^{(1,0)} - \ket{\varphi_t} \bra{\varphi_t}) \sqrt{1 - \Delta_1})  
 \\
\label{eq: Nelson energy trace 1}
&= \SCP{\Psi_N}{p^{\varphi_t}_1 \sqrt{1- \Delta_1} A_1 \sqrt{1 - \Delta_1} p^{\varphi_t}_1 \Psi_N}
- \scp{\varphi_t}{\sqrt{1 - \Delta_1} A_1 \sqrt{1 - \Delta_1} \varphi_t} 
\\
\label{eq: Nelson energy trace 2}
&+  \SCP{\Psi_N}{q^{\varphi_t}_1 \sqrt{1- \Delta_1} A_1 \sqrt{1 - \Delta_1} p^{\varphi_t}_1 \Psi_N}
+ \SCP{\Psi_N}{p^{\varphi_t}_1 \sqrt{1- \Delta_1} A_1 \sqrt{1 - \Delta_1} q^{\varphi_t}_1 \Psi_N} 
\\
\label{eq: Nelson energy trace 3}
&+ \SCP{\Psi_N}{q^{\varphi_t}_1 \sqrt{1- \Delta_1} A_1 \sqrt{1 - \Delta_1} q^{\varphi_t}_1 \Psi_N}.
\end{align}
By means of
\begin{align}
\norm{\sqrt{1 - \Delta_1} q^{\varphi_t}_1 \Psi_N}^2 
=& \norm{q^{\varphi_t}_1 \Psi_N}^2 + \norm{\nabla_1 q^{\varphi_t}_1 \Psi_N}^2 
\leq \beta^a(t) +  \beta^c(t)
\end{align}
and
\begin{align}
\norm{\sqrt{1 - \Delta_1} p^{\varphi_t}_1}_{op}^2 
\leq& \scp{\varphi_t}{\left( 1 - \Delta_1 \right) \varphi_t}
= \norm{\varphi_t}_{H^1(\mathbb{R}^3)}^2
\end{align}
we estimate
\begin{align}\abs{\eqref{eq: Nelson energy trace 1}}
\leq& \abs{\scp{\varphi_t}{\sqrt{1 - \Delta} A_1 \sqrt{1 - \Delta} \varphi_t}} \abs{\SCP{\Psi_N}{p^{\varphi_t}_1 \Psi_N} -1}
\leq \norm{A_1}_{op} \norm{\varphi_t}^2_{H^1(\mathbb{R}^3)} \beta^a(t), 
\nonumber \\
\abs{\eqref{eq: Nelson energy trace 2}}
\leq& 2   \norm{A_1}_{op} \norm{\varphi_t}_{H^1(\mathbb{R}^3)} \sqrt{\beta^a(t) + \beta^c(t)}, 
\nonumber \\
\abs{\eqref{eq: Nelson energy trace 3}}
\leq& \norm{A_1}_{op} \left( \beta^a(t) +\beta^c(t) \right).
\end{align}
This leads to 
 \begin{align}
\text{Tr}_{L^2(\mathbb{R}^3)} \abs{ \sqrt{1- \Delta} \left(\gamma_{N,t}^{(1,0)} - \ket{\varphi_t} \bra{\varphi_t} \right) \sqrt{1 - \Delta}}  \leq&
\left(1 + \norm{\varphi_t}_{H^1(\mathbb{R}^3)}^2 \right)   \left( \beta^a(t) + \beta^c(t)  \right)   
\nonumber \\
+& 2 \norm{\varphi_t}_{H^1(\mathbb{R}^3)} \sqrt{\beta^a(t) + \beta^c(t)}.
 \end{align}
\end{proof}

\section{Estimates on the time derivative}
\label{sec: Estimates on the time derivative}

\subsection{Preliminary estimates}

In the following, we control the change of $\beta$ in time by separately estimating the time derivative of $\beta^a$ and $\beta^b$. On the one hand a change in $\beta^a$ is caused by the fraction of particles which are not in the condensate state $\varphi_t$. This behavior is analogous to the growth of diseases, where the infection rate of cells (or particles that will leave the condensate) at a given time is proportional to the number of already infected cells. On the other hand there will be a change due to the fact that the particles of the many-body system couple to the quantized radiation field, whereas the condensate wave function is in interaction with the classical field. 
To control the difference between the quantized and classical field by the functional $\beta^b$ we will have to split the radiation fields in their positive and negative frequency parts.
\begin{align}
\label{eq: Nelson definition positive and negative frequency part}
\vphp(x) &\coloneqq \int d^3k \, \frac{\tilde{\kappa}(k)}{\sqrt{2 \omega(k)}} e^{ikx} a(k) ,  \quad \; \; 
\vphm(x)  \coloneqq \int d^3k \, \frac{\tilde{\kappa}(k)}{\sqrt{2 \omega(k)}} e^{-ikx} a^*(k) ,  
\nonumber\\
\vphpc(x,t) &\coloneqq \int d^3k \, \frac{\tilde{\kappa}(k)}{\sqrt{2 \omega(k)}} e^{ikx} \alpha_t(k) ,  \;
\vphmc(x,t) \coloneqq \int d^3k \, \frac{\tilde{\kappa}(k)}{\sqrt{2 \omega(k)}} e^{-ikx} \alpha_t^*(k) .  
\end{align}
For technical reason it is then helpful to introduce the following (less singular) auxiliary fields
\begin{align}
\label{eq: Nelson auxilliary fields}
\Fp(x) &\coloneqq \int d^3k \, \tilde{\kappa}(k) e^{ikx} a(k) ,  \quad \; \; 
\Fm(x) \coloneqq \int d^3k \, \tilde{\kappa}(k) e^{-ikx} a^*(k) ,  
\nonumber \\
\Fpc(x,t) &\coloneqq \int d^3k \, \tilde{\kappa}(k) e^{ikx} \alpha_t(k) ,  \;
\Fmc(x,t) \coloneqq \int d^3k \, \tilde{\kappa}(k) e^{-ikx} \alpha_t^*(k) .  
\end{align}
By means of the cutoff function
\begin{equation}
\label{eq: cutoff function 2}
\tilde{\eta}(k) \coloneqq \frac{\tilde{\kappa}(k)}{\sqrt{2 \omega(k)}} 
= \frac{(2 \pi)^{-3/2}}{\sqrt{2 \omega(k)}} \id_{\abs{k} \leq \Lambda}(k)
\end{equation}
we are able to express the scalar fields in terms of the auxiliary fields.
\begin{lemma}
\label{lemma: auxially fields and cutoff}
Let $\eta$ be the Fourier transform of \eqref{eq: cutoff function 2}, then
\begin{align}
\vphp(x) =& \left( \eta * \Fp \right)(x) ,
\quad \; \;
\vphm(x) = \left( \eta * \Fm \right)(x) , 
\nonumber\\
\vphpc(x,t) =& \left( \eta * \Fpc \right)(x,t) ,
\;
\vphmc(x,t) = \left( \eta * \Fmc \right)(x,t).
\end{align}
\end{lemma}
\begin{proof}
The proof is a simple application of convolutions theorem.
\end{proof}
\noindent
In the following, we will integrate the form-factor $\eta$ of the radiation field and estimate the difference in the auxiliary fields. This requires that the $L^2$-norms of the cutoff functions 
\begin{align}
\label{eq: Nelson cutoff functions norm}
\norm{\kappa}_2^2 = \Lambda^3/(6 \pi^2)
\quad \text{and} \quad
\norm{\eta}_2^2 \leq \Lambda^2/(4 \pi^2)
\end{align}
are finite. Subsequently, we use Plancherel's theorem and estimate the difference in the positive frequency parts of the auxiliary fields by 
\begin{align}
&\int d^3y \, \norm{\left( N^{-1/2} \Fp(y) - \Fpc(y,t) \right) \Psi_{N,t}}^2 
= \int d^3k \, \norm{\left( N^{-1/2} \Fp(k) - \Fpc(k,t) \right) \Psi_{N,t}}^2 
\nonumber \\
= &\int_{\abs{k} \leq \Lambda} d^3k \,
\SCP{\left( N^{-1/2} a(k) - \alpha_t(k) \right) \Psi_{N,t}}{\left( N^{-1/2} a(k) - \alpha_t(k) \right) \Psi_{N,t}}
\leq \beta^b\left( \Psi_{N,t}, \alpha_t\right).
\end{align}
Pulling the pieces together we get
\begin{lemma}
\label{lemma: Nelson field difference versus beta-b}
Let $\alpha_t \in L^2(\mathbb{R}^3)$ and $\Psi_{N,t} \in \mathcal{H}^{(N)} \cap \mathcal{D}\left(\mathcal{N} \right)$. Then, there exists a generic constant $C$ independent of $N$, $\Lambda$ and $t$ such that
\begin{align}
\norm{\left( N^{-1/2} \vph(x_1) - \vphc(x_1,t) \right) \Psi_{N,t}}^2 
\leq& C \Lambda^2 \left( \beta^b\left( \Psi_{N,t},\alpha_t \right) + N^{-1} \right),  \\
\norm{\left( N^{-1/2} \vphm(x_1) - \vphmc(x_1,t) \right) \Psi_{N,t}}^2 
\leq& C \Lambda^2 \left( \beta^b\left( \Psi_{N,t},\alpha_t \right) + N^{-1} \right),  \\
\norm{\left( N^{-1/2} \vphp(x_1) - \vphpc(x_1,t) \right) p_1 \Psi_{N,t}}^2 
\leq& C \Lambda^2 \beta^b\left( \Psi_{N,t},\alpha_t \right). 
\end{align}
\end{lemma}

\begin{proof}
From the canonical commutation relations \eqref{eq: Nelson canonical commutation relation}, we obtain
\begin{align}
\left[ \left( N^{-1/2} \vphp(x) - \vphp(x,t) \right),  \left( N^{-1/2} \vphm(x) - \vphm(x,t) \right) \right]
=&    N^{-1} \norm{\eta}_2^2
\end{align}
and estimate
\begin{align}
&\norm{\left( N^{-1/2} \vph(x_1) - \vphc(x_1,t) \right) \Psi_{N}}^2 
\nonumber \\
\leq  2 &\norm{\left( N^{-1/2} \vphp(x_1) - \vphpc(x_1,t) \right) \Psi_{N}}^2 
+ 2 \norm{\left( N^{-1/2} \vphm(x_1) - \vphmc(x_1,t) \right) \Psi_{N}}^2 
\nonumber \\
\leq 4 &\norm{\left( N^{-1/2} \vphp(x_1) - \vphpc(x_1,t) \right) \Psi_{N}}^2 
+ 2 N^{-1} \norm{\eta}_2^2.
\end{align}
By means of Lemma~\ref{lemma: auxially fields and cutoff}  we have
\small
\begin{align}
&\norm{\left( N^{-1/2} \vphp(x_1)  - \vphpc(x_1,t) \right) \Psi_N}^2  
\nonumber \\
= &\SCP{\int d^3y \, \eta(x_1-y) \left( N^{-1/2} \Fp(y) - \Fpc(y,t)\right)  \Psi_N}{\int d^3z \, \eta(x_1-z) \left( N^{-1/2} \Fp(z) - \Fpc(z,t)\right)  \Psi_N}  
\nonumber \\
\leq& \int d^3y \int d^3z \,
\abs{\SCP{\eta^*(x_1-z) \left( N^{-1/2} \Fp(y) - \Fpc(y,t)\right)  \Psi_N}{ \eta^*(x_1-y) \left( N^{-1/2} \Fp(z) - \Fpc(z,t)\right)  \Psi_N}}.
\end{align}
\normalsize
Cauchy-Schwarz inequality and the estimate $ ab \leq 1/2 \left(a^2 + b^2 \right)$ give rise to
\small
\begin{align}
&\norm{\left( N^{-1/2} \vphp(x_1) - \vphpc(x_1,t) \right) \Psi_N}^2 
\nonumber \\
\leq& \int d^3y \int d^3z \, \norm{\eta^*(x_1-z) \left( N^{-1/2} \Fp(y) - \Fpc(y,t)\right)  \Psi_N}
\norm{\eta^*(x_1-y) \left( N^{-1/2} \Fp(z) - \Fpc(z,t)\right)  \Psi_N} 
\nonumber \\
\leq& \int d^3y \int d^3z \, \norm{\eta^*(x_1-z) \left( N^{-1/2} \Fp(y) - \Fpc(y,t)\right)  \Psi_N}^2 
\nonumber \\
=& \int d^3y \, \SCP{\left( N^{-1/2} \Fp(y) - \Fpc(y,t)\right)  \Psi_N}{ \int d^3z \abs{\eta(x_1-z)}^2  \left( N^{-1/2} \Fp(y) - \Fpc(y,t)\right)  \Psi_N} 
\nonumber \\
=& \norm{\eta}_2^2 \int d^3y \, \norm{\left( N^{-1/2} \Fp(y) - \Fpc(y,t)\right)  \Psi_N}^2 
\leq \norm{\eta}_2^2 \beta^b\left( \Psi_{N,t} , \alpha_t \right)  .
\end{align}
\normalsize
In total, we get
\begin{align}
\norm{\left( N^{-1/2} \vph(x_1) - \vphc(x_1,t) \right) \Psi_{N,t}}^2   
&\leq  \norm{\eta}_2^2 \left( 4 \beta^b\left( \Psi_{N,t} , \alpha_t \right) + 2 N^{-1} \right)
\nonumber \\
&\leq C \Lambda^2 \left( \beta^b\left( \Psi_{N,t} , \alpha_t \right) + N^{-1} \right).
\end{align}
The second and third inequality are shown analogously. Hereby, it is helpful to recall that $\left[p_1, \Fp(y) \right] = \left[p_1, \Fpc(y) \right] = 0$.
\end{proof} 

\subsection{Estimate on the time derivative of $\beta$}
\label{subsec: Estimate on the time derivative of beta}

Subsequently, we control the change of $\beta \left( \Psi_{N,t}, \varphi_t, \alpha_t \right)$ in time.
 
\begin{lemma}
\label{lemma: Nelson time derivative of beta preparation}
Let $(\varphi_t,\alpha_t) \in \mathcal{G}_1$ and $\Psi_{N,t}$ be the unique solution of \eqref{eq: Nelson Schroedinger equation microscopic} with initial data $ \Psi_{N,0} \in \left( L_s^2(\mathbb{R}^{3N})  \otimes \mathcal{F} \right) \cap \mathcal{D} \left( \mathcal{N} \right) \cap \mathcal{D} \left(  \mathcal{N}  H_N \right)$ such that $\norm{\Psi_{N,0}} =1$. 
Then
\begin{align}
d_t \beta^a(t)
&= - 2 \Im \SCP{\Psi_{N,t}}{\left( N^{-1/2} \vph(x_1) - \vphc(x_1,t) \right) q_1^{\varphi_t} \Psi_{N,t}},
\nonumber \\
\label{eq: Nelson time derivative of beta preparation}
d_t \beta^b(t) &= \; \;  2 \Im \SCP{\Psi_{N,t}}{\Big(  \int d^3k \, \tilde{\eta}(k) (2 \pi)^{3/2} \mathcal{FT}^*[\abs{\varphi_t}^2](k) \left( N^{-1/2} a(k) - \alpha_t(k) \right)   \Big) \Psi_{N,t}}
\nonumber\\
 & \; \;  - 2 \Im \SCP{\Psi_{N,t}}{ \Big( \int d^3k \, \tilde{\eta}(k) e^{ikx_1} \left(  N^{-1/2} a(k) - \alpha_t(k) \right)  \Big) \Psi_{N,t}}.
\end{align}
\end{lemma} 

\begin{proof}
The structure of the proof is best understood by the following formal calculation. 
A rigorous derivation which requires to show the invariance of the domain $\mathcal{D} \left( \mathcal{N} \right) \cap \mathcal{D} \left(  \mathcal{N}  H_N \right)$ during the time evolution is presented in \cite[Appendix 2.11]{leopold2}.\\
The functional $\beta^a(t)$ is time-dependent, because $\Psi_{N,t}$ and $\varphi_t$ evolve according to \eqref{eq: Nelson Schroedinger equation microscopic} and \eqref{eq: Nelson Schroedinger-Klein-Gordon system} respectively. The derivative of the projector $q^{\varphi_t}$ is given by
\begin{align}
d_t q_1^{\varphi_t} = - i \left[H_1^{eff},q_1^{\varphi_t} \right],
\end{align}
where $H_1^{eff} = - \Delta_1 + \vphc(x_1,t)$ is the effective Hamiltonian acting on the first variable. This leads to
\begin{align}
d_t \beta^a(t) = d_t \SCP{\Psi_{N,t}}{q_1^{\varphi_t} \Psi_{N,t}} 
&= \qquad \,   i \SCP{\Psi_{N,t}}{\left[ \left( H_N - H_1^{eff} \right) , q_1^{\varphi_t}  \right] \Psi_{N,t}}
\nonumber \\  
&= \qquad \,   i \SCP{\Psi_{N,t}}{\left[ \left( N^{-1/2}\vph(x_1) - \vphc(x_1,t) \right), q_1^{\varphi_t} \right] \Psi_{N,t}}  
\nonumber \\
&= - 2 \Im \SCP{\Psi_{N,t}}{\left( N^{-1/2} \vph(x_1) - \vphc(x_1,t) \right) q_1^{\varphi_t} \Psi_{N,t}}.
\end{align}
We calculate the commutator
\begin{align}
 i \left[ H_N , \left( N^{-1/2} a(k) - \alpha_t(k)  \right)  \right] 
 =& - i  \omega(k) N^{-1/2} a(k) - i N^{-1} \sum_{j=1}^N \tilde{\eta}(k) e^{- i k x_j}    
\end{align}
by means of the canonical commutation relations~\eqref{eq: Nelson canonical commutation relation} and continue with
\begin{align}
d_t \beta^b(t) 
&=  \int d^3k \, d_t \SCP{\left( N^{-1/2} a(k) - \alpha_t(k) \right) \Psi_{N,t}}{   \left( N^{-1/2} a(k) - \alpha_t(k) \right) \Psi_{N,t}} 
\nonumber \\
&=  \int d^3k \, \SCP{ i \left[ H_N, \left( N^{-1/2} a(k) - \alpha_t(k)  \right)  \right] \Psi_{N,t}}{ \left( N^{-1/2} a(k) - \alpha_t(k) \right) \Psi_{N,t}}
\nonumber \\
&+   \int d^3k \,  \SCP{\left( N^{-1/2} a(k) - \alpha_t(k) \right) \Psi_{N,t}}{ i \left[ H_N , \left(  N^{-1/2} a(k) - \alpha_t(k) \right) \right]  \Psi_{N,t}}  
\nonumber \\
&- \int d^3k \, \SCP{\left( \partial_t \alpha_t \right)(k) \Psi_{N,t}}{ \left( N^{-1/2} a(k) - \alpha_t(k) \right) \Psi_{N,t}}
\nonumber \\
&- \int d^3k \, \SCP{\left( N^{-1/2} a(k) - \alpha_t(k) \right) \Psi_{N,t}}{ \left( \partial_t \alpha_t \right)(k) \Psi_{N,t}}
\nonumber\\
&=  2 \int d^3k \, \Re \SCP{ i \left[ H_N, \left( N^{-1/2} a(k) - \alpha_t(k)  \right)  \right] \Psi_{N,t}}{ \left( N^{-1/2} a(k) - \alpha_t(k) \right) \Psi_{N,t}}
\nonumber \\
&- 2 \int d^3k \, \Re \SCP{\left( \partial_t \alpha_t \right)(k) \Psi_{N,t}}{ \left( N^{-1/2} a(k) - \alpha_t(k) \right) \Psi_{N,t}}
\nonumber \\
&= 2 \int d^3k \, \Re \big\{  i \omega(k) \SCP{ \left( N^{-1/2} a(k) - \alpha_t(k) \right) \Psi_{N,t}}{ \left( N^{-1/2} a(k) - \alpha_t(k) \right) \Psi_{N,t}}    \big\}
\nonumber\\
&+ 2 \int d^3k \, \Re \big\{  i  \SCP{ N^{-1}\sum_{j=1}^N \tilde{\eta}(k) e^{- ik x_j} \Psi_{N,t}}{ \left( N^{-1/2} a(k) - \alpha_t(k) \right) \Psi_{N,t}}    \big\}
\nonumber\\
&- 2 \int d^3k \, \Re \big\{  i  \SCP{ (2 \pi)^{3/2} \tilde{\eta}(k) \mathcal{FT}[\abs{\varphi_t}^2](k) \Psi_{N,t}}{ \left( N^{-1/2} a(k) - \alpha_t(k) \right) \Psi_{N,t}}    \big\}.
\end{align}
So if we use the symmetry of the wave function and $\Re \{ i z\} = - \Im \{ z \} $, we get
\begin{align}
d_t \beta^b(t) 
&= - 2 \int d^3k \, \Im \big\{   \omega(k) \SCP{ \left( N^{-1/2} a(k) - \alpha_t(k) \right) \Psi_{N,t}}{ \left( N^{-1/2} a(k) - \alpha_t(k) \right) \Psi_{N,t}}    \big\}
\nonumber\\
&- 2 \int d^3k \, \Im \big\{   \SCP{ \tilde{\eta}(k) e^{- ik x_1} \Psi_{N,t}}{ \left( N^{-1/2} a(k) - \alpha_t(k) \right) \Psi_{N,t}}    \big\}
\nonumber\\
&+ 2 \int d^3k \, \Im \big\{    \SCP{ (2 \pi)^{3/2} \tilde{\eta}(k) \mathcal{FT}[\abs{\varphi_t}^2](k) \Psi_{N,t}}{ \left( N^{-1/2} a(k) - \alpha_t(k) \right) \Psi_{N,t}}    \big\}
\nonumber\\
&= 2  \Im \big\{  \SCP{ \Psi_{N,t}}{ \Big( \int d^3k \, (2 \pi)^{3/2} \tilde{\eta}(k) \mathcal{FT}^*[\abs{\varphi_t}^2](k) \left( N^{-1/2} a(k) - \alpha_t(k) \right)  \Big) \Psi_{N,t}}    \big\}
\nonumber\\
&- 2  \Im \big\{   \SCP{  \Psi_{N,t}}{ \Big( \int d^3k \,  \tilde{\eta}(k) e^{ ik x_1} \left( N^{-1/2} a(k) - \alpha_t(k) \right) \Big) \Psi_{N,t}}    \big\}.
\end{align}

\end{proof}

\begin{lemma}
\label{lemma: Nelson time derivative of beta}
Let $(\varphi_t,\alpha_t) \in \mathcal{G}_1$ and $\Psi_{N,t}$ be the unique solution of \eqref{eq: Nelson Schroedinger equation microscopic} with initial data $ \Psi_{N,0} \in \left( L_s^2(\mathbb{R}^{3N})  \otimes \mathcal{F} \right) \cap \mathcal{D} \left( \mathcal{N} \right) \cap \mathcal{D} \left(  \mathcal{N}  H_N \right)$ such that $\norm{\Psi_{N,0}} =1$. 
Then for any $t \in \mathbb{R}_0^+$ there exists a generic constant $C$ independent of $N$, $\Lambda$ and $t$ such that
\begin{align}
\abs{d_t \beta \left( \Psi_{N,t}, \varphi_t, \alpha_t \right) } \leq& C \Lambda^2 \left( \beta \left( \Psi_{N,t}, \varphi_t, \alpha_t \right)  + N^{-1} \right), \\
\label{eq: Nelson time derivative of beta}
\beta \left( \Psi_{N,t}, \varphi_t, \alpha_t \right) 
\leq& e^{C \Lambda^2 t} \left( \beta \left( \Psi_{N,0}, \varphi_0, \alpha_0 \right) + N^{-1} \right).
\end{align}
\end{lemma}

\begin{proof}
Schwarz's inequality and  $ab \leq 1/2 (a^2 + b^2)$ let us estimate the first line of Lemma~\ref{lemma: Nelson time derivative of beta preparation} by
\begin{align}
\abs{d_t \beta^a  \left(t \right)} 
&\leq 2 \abs{\SCP{\Psi_{N,t}}{\left( N^{-1/2} \vph(x_1)  - \vphc(x_1,t) \right) q_1^{\varphi_t} \Psi_{N,t}}}
\nonumber \\
&\leq \norm{\left( N^{-1/2} \vph(x_1)  - \vphc(x_1,t) \right) \Psi_{N,t}}^2 +  
\norm{q_1^{\varphi_t} \Psi_{N,t}}^2.
\end{align}
By Lemma~\ref{lemma: Nelson field difference versus beta-b}, we obtain
\begin{align}
\abs{d_t \beta^a  \left( t \right)} 
&\leq  C \Lambda^2 \left( \beta\left( t \right) + N^{-1} \right).
\end{align}
In order to estimate $d_t \beta^b(t)$ we notice that
\begin{align}
\int d^3k \, \tilde{\eta}(k) e^{ikx_1} \left( N^{-1/2} a(k) - \alpha(k,t) \right)
&= \int d^3y \,  \eta(x_1 - y)  \left(
N^{-1/2} \Fp(y) - \Fpc(y,t) \right)(x_1)
\nonumber \\
&=  \quad \; \left(  N^{-1/2} \vphp(x_1) - \vphpc(x_1,t) \right) 
\end{align}
and
\begin{align}
\int d^3k \, &\tilde{\eta}(k) (2\pi)^{3/2} \mathcal{FT}[\abs{\varphi_t}^2]^*(k) \left( N^{-1/2} a(k) - \alpha_t(k) \right)
\nonumber \\
&= \int d^3y \, \left( \eta * \abs{\varphi_t}^2 \right)(y,t) \left(
N^{-1/2} \Fp(y) - \Fpc(y,t) \right).
\end{align}
follow from the convolution theorem.
This gives
\begin{align}
d_t \beta^b(t) 
= - 2 \Im &\int d^3y \, \SCP{\Psi_{N,t}}{\eta(x_1 - y) \left(  
N^{-1/2} \Fp(y) - \Fpc(y,t) \right) \Psi_{N,t}} 
\nonumber \\
+ 2 \Im &\int d^3y \, \SCP{\Psi_{N,t}}{ \left( \eta * \abs{\varphi_t}^2 \right)(y,t) \left( N^{-1/2} \Fp(y) - \Fpc(y,t) \right) \Psi_{N,t}} .
\end{align}
We see that not only present gauge boson fluctuations around the coherent state lead to a growth in $\beta^b(t)$ but an additional change appears, because the second quantized radiation field couples to the mean particle density of the many-body system while the source of the classical field is given by the density of the condensate wave function. In order to estimate the difference between the densities by the functional $\beta^a(t)$ we insert the identity $ 1 = p^{\varphi_t}_1 + q^{\varphi_t}_1$. 
\begin{align}
d_t \beta^b(t) 
= - 2 \Im  &\int d^3y \, \SCP{\Psi_{N,t}}{p^{\varphi_t}_1 \eta(x_1 - y) p^{\varphi_t}_1 \left( N^{-1/2} \Fp(y) - \Fpc(y,t) \right) \Psi_{N,t}} 
\nonumber \\
+ 2 \Im  &\int d^3y \, \SCP{\Psi_{N,t}}{ \left( \eta * \abs{\varphi_t}^2 \right)(y,t) \left( N^{-1/2} \Fp(y) - \Fpc(y,t) \right) \Psi_{N,t}}
\nonumber \\
-2  \Im  &\int d^3y \, \SCP{\Psi_{N,t}}{q^{\varphi_t}_1 \eta(x_1 - y) p^{\varphi_t}_1 \left( N^{-1/2} \Fp(y) - \Fpc(y,t) \right) \Psi_{N,t}}
\nonumber \\
-2  \Im  &\int d^3y \,  \SCP{\Psi_{N,t}}{ \eta(x_1 - y) q^{\varphi_t}_1  \left( N^{-1/2} \Fp(y) - \Fpc(y,t) \right) \Psi_{N,t}} .
\end{align}
The first two lines are the most important. They become small, because the mean particle density of the many-body system is approximately given by the density of the condensate wave function. From $\eta(-x) = \eta(x)$ we conclude
\begin{align}
p^{\varphi_t}_1 \eta(x_1-y) p^{\varphi_t}_1 
=& p^{\varphi_t}_1 \int d^3z \, \eta(z-y) \abs{\varphi_t}^2(z,t) = p^{\varphi_t}_1 \left( \eta * \abs{\varphi_t}^2 \right)(y,t)
\end{align}
and continue with
\begin{align}
d_t \beta^b(t) 
= -2 &\Im \int d^3y \, \SCP{\Psi_{N,t}}{(p^{\varphi_t}_1 -1 ) \left( \eta * \abs{\varphi_t}^2 \right)(y,t)  \left( N^{-1/2} \Fp(y) - \Fpc(y,t) \right) \Psi_{N,t}}  
\nonumber \\
-2  &\Im \SCP{\Psi_{N,t}}{q_1 \int d^3y \, \eta(x_1 - y)  \left( N^{-1/2} \Fp(y) - \Fpc(y,t) \right) p^{\varphi_t}_1 \Psi_{N,t}} 
\nonumber \\
-2  &\Im \SCP{\Psi_{N,t}}{ \int d^3y \, \eta(x_1 - y)  \left( N^{-1/2} \Fp(y) - \Fpc(y,t) \right) q^{\varphi_t}_1 \Psi_{N,t}}
\nonumber \\
\label{eq: Nelson beta-b-1}
= \; \; 2 &\Im \int d^3y \, \SCP{\Psi_{N,t}}{ q^{\varphi_t}_1 \left( \eta * \abs{\varphi_t}^2 \right)(y,t)  \left( N^{-1/2} \Fp(y) - \Fpc(y,t) \right) \Psi_{N,t}}  
\\
\label{eq: Nelson beta-b-2}
-2  &\Im \SCP{\Psi_{N,t}}{q^{\varphi_t}_1  \left( N^{-1/2} \vphp(x_1) - \vphpc(x_1,t) \right) p^{\varphi_t}_1 \Psi_{N,t}} 
\\
\label{eq: Nelson beta-b-3}
-2  &\Im \SCP{\Psi_{N,t}}{  \left( N^{-1/2} \vphp(x_1) - \vphpc(x_1,t) \right) q^{\varphi_t}_1 \Psi_{N,t}}. 
\end{align}
In the following, we estimate each line separately.
\begin{align}
\abs{\eqref{eq: Nelson beta-b-1}}
&\leq 2 \abs{ \int d^3y \, \SCP{\left( \eta * \abs{\varphi_t}^2 \right)(y,t) q^{\varphi_t}_1 \Psi_{N,t}}{\left( N^{-1/2} \Fp(y) - \Fpc(y,t) \right) \Psi_{N,t}}} 
\nonumber \\
&\leq \; \; \;  \int d^3y \, \SCP{q^{\varphi_t}_1 \Psi_{N,t}}{\abs{ \left(\eta * \abs{\varphi_t}^2 \right)(y,t)}^2 q^{\varphi_t}_1 \Psi_N}  
\nonumber\\
&+  \; \; \;   \int d^3y \, \norm{\left( N^{-1/2} \Fp(y) - \Fpc(y,t) \right) \Psi_{N,t}}^2 
\nonumber  \\
&\leq  \; \;\; \norm{ \eta * \abs{\varphi_t}^2}_2^2 \SCP{\Psi_{N,t}}{q^{\varphi_t}_1 \Psi_{N,t}} + \beta^b(t)
\leq C \Lambda^2 \beta(t).
\end{align}
Here we have used that
\begin{align}
\norm{ \eta * \abs{\varphi_t}^2}_2
\leq& \norm{\eta}_2 \norm{\abs{\varphi_t}^2}_1 = \norm{\eta}_2 \norm{\varphi_t}_2^2 = C \Lambda 
\end{align}
holds due to Young's inequality and \eqref{eq: Nelson cutoff functions norm}.
Lemma~\ref{lemma: Nelson field difference versus beta-b} leads to
\begin{align}
\abs{\eqref{eq: Nelson beta-b-2}}
\leq& 2 \abs{\SCP{q^{\varphi_t}_1 \Psi_N}{  \left( N^{-1/2} \vphp(x_1)  - \vphpc(x_1,t) \right) p^{\varphi_t}_1 \Psi_N}} 
\nonumber \\
\leq& \;  \norm{\left( N^{-1/2} \vphp(x_1)  - \vphpc(x_1,t) \right) p^{\varphi_t}_1 \Psi_N}^2
+ \norm{q^{\varphi_t}_1 \Psi_N}^2 
\leq C \Lambda^2 \beta(t)
\end{align}
and
\begin{align}
\abs{\eqref{eq: Nelson beta-b-3}} 
\leq& 2 \abs{ \SCP{\left( N^{-1/2} \vphm(x_1) - \vphmc(x_1,t) \right) \Psi_N}{   q^{\varphi_t}_1 \Psi_N}}
\nonumber \\
\leq& \; \norm{\left( N^{-1/2} \vphm(x_1) - \vphmc(x_1,t) \right) \Psi_N}^2
+ \norm{q^{\varphi_t}_1 \Psi_N}^2
\nonumber\\
\leq& C \Lambda^2 \left( \beta(t) + N^{-1} \right).
\end{align}
In total we have
\begin{align}
\abs{d_t \beta^b \left( t \right)}
\leq& C \Lambda^2 \left(  \beta \left( t \right) + N^{-1}  \right).
\end{align}
Now we can put the terms together to get
\begin{align}
d_t \beta \left( t \right) &\leq \abs{d_t \beta^a \left( t \right)} + \abs{d_t \beta^b \left( t \right)}
\leq C \Lambda^2 \left( \beta \left( t \right) + N^{-1}  \right).
\end{align}
Applying Gronwall's lemma proves
\begin{align}
\beta\left( t \right)
\leq& e^{C \Lambda^2 t} \left( \beta\left( 0 \right) + N^{-1} \right).
\end{align}

\end{proof}

\subsection{Control of the kinetic energy}
\label{subsec:: Control of the kinetic energy}

In order to prove the convergence of the one-particle reduced density matrix of the charges in Sobolev norm it is necessary to control the kinetic energy of the particles which are not in the condensate (see Section~\ref{sec: Nelson Relation to reduced density matrices}). To this end we include $\beta^c(\Psi_{N,t},\varphi_t) \coloneqq \norm{\nabla_1 q^{\varphi_t}_1 \Psi_{N,t}}^2$ in the definition of the functional and perform a Gronwall estimate for the redefined functional $\beta_2(\Psi_{N,t},\varphi_t, \alpha_t)$.

\begin{lemma}
\label{lemma: Nelson time derivative of beta-c preparation}
Let $(\varphi_t,\alpha_t) \in \mathcal{G}_2$ and $\Psi_{N,t}$ be the unique solution of \eqref{eq: Nelson Schroedinger equation microscopic} with initial data $ \Psi_{N,0} \in \left( L_s^2(\mathbb{R}^{3N})  \otimes \mathcal{F} \right) \cap \mathcal{D} \left( \mathcal{N} \right) \cap \mathcal{D} \left(  \mathcal{N}  H_N \right) \cap \mathcal{D}  \left( H_N^2  \right)$ such that $\norm{\Psi_{N,0}} =1$. 
Then
\begin{align}
\label{eq: Nelson time derivative of beta-c preparation}
d_t \beta^c(\Psi_{N,t},\varphi_t)
&= 2 \Im \SCP{p^{\varphi_t}_1 \left( N^{-1/2} \vph(x_1) - \vphc(x_1,t) \right) \Psi_{N,t}}{(- \Delta_1) q^{\varphi_t}_1 \Psi{N,t}}
\nonumber\\
&- 2 \Im \SCP{\left( N^{-1/2} \vph(x_1) - \vphc(x_1,t) \right) p^{\varphi_t}_1 \Psi_{N,t}}{ (- \Delta_1) q_1^{\varphi_t} \Psi_{N,t}}
\nonumber\\
&- 2 \Im \SCP{N^{-1/2} \vph(x_1) q^{\varphi_t}_1 \Psi_{N,t}}{(- \Delta_1) q^{\varphi_t}_1 \Psi_{N,t}} .
\end{align}
\end{lemma} 

\begin{proof}
We infer $\Psi_{N,t} \in \left( L_s^2(\mathbb{R}^{3N})  \otimes \mathcal{F} \right) \cap \mathcal{D} \left( \mathcal{N} \right) \cap \mathcal{D} \left(  \mathcal{N}  H_N \right) \cap \mathcal{D}  \left( H_N^2  \right)$ for all $t \in \mathbb{R}_0^+$ from  $\Psi_{N,0} \in \left( L_s^2(\mathbb{R}^{3N})  \otimes \mathcal{F} \right) \cap \mathcal{D} \left( \mathcal{N} \right) \cap \mathcal{D} \left(  \mathcal{N}  H_N \right) \cap \mathcal{D}  \left( H_N^2  \right)$ 
by Stone's Theorem and the invariance of $\mathcal{D} \left( \mathcal{N} \right) \cap \mathcal{D} \left(  \mathcal{N}  H_N \right)$ during the time evolution (see \cite[Appendix 2.11]{leopold2}).
 This ensures that the following expressions are well defined. The derivative of $\beta^c(t)$ is determined by
\begin{align}
d_t \beta^c(t) 
&= \qquad \, i \SCP{q^{\varphi_t}_1  H_N \Psi_{N,t}}{ (- \Delta_1) q^{\varphi_t}_1 \Psi_{N,t}}
 \qquad \;  - i \SCP{q^{\varphi_t}_1 \Psi_{N,t}}{(- \Delta_1) q^{\varphi_t}_1  H_N \Psi_{N,t}}
\nonumber\\
& \qquad \; + i \SCP{ \left[ H_1^{eff}, q^{\varphi_t}_1 \right] \Psi_{N,t}}{ (- \Delta_1) q^{\varphi_t}_1 \Psi_{N,t}}
- i \SCP{q^{\varphi_t}_1 \Psi_{N,t}}{(- \Delta_1)  \left[ H_1^{eff}, q^{\varphi_t}_1  \right] \Psi_{N,t}}
\nonumber\\
&= \qquad \,  i \SCP{q^{\varphi_t}_1  H_N \Psi_{N,t}}{ (- \Delta_1) q^{\varphi_t}_1 \Psi_{N,t}}
 \qquad \;   - i \SCP{(- \Delta_1) q^{\varphi_t}_1 \Psi_{N,t}}{ q^{\varphi_t}_1  H_N \Psi_{N,t}}
\nonumber\\
& \qquad \;  +i  \SCP{ \left[ H_1^{eff}, q^{\varphi_t}_1 \right] \Psi_{N,t}}{ (- \Delta_1) q^{\varphi_t}_1 \Psi_{N,t}}
-i  \SCP{(- \Delta_1) q^{\varphi_t}_1 \Psi_{N,t}}{  \left[ H_1^{eff}, q^{\varphi_t}_1  \right] \Psi_{N,t}}
\nonumber\\
&= -2 \Im \SCP{q^{\varphi_t}_1 H_N \Psi_{N,t}}{ (- \Delta_1) q^{\varphi_t}_1 \Psi_{N,t}}
\nonumber\\
 & \; \; \, \,  -2 \Im \SCP{\left[ H_1^{eff}, q^{\varphi_t}_1 \right] \Psi_{N,t}}{ (- \Delta_1) q^{\varphi_t}_1 \Psi_{N,t}}.
\end{align}
Since
$
\SCP{q^{\varphi_t}_1 \left( - \Delta_i + N^{-1/2} \vph(x_i) \right) \Psi_{N,t}}{ (- \Delta_1 ) q^{\varphi_t}_1 \Psi_{N,t}} 
$ and $\SCP{q^{\varphi_t}_1 H_f \Psi_{N,t}}{(- \Delta_1) q^{\varphi_t}_1 \Psi_{N,t}}$ are real numbers for $i \in \{ 2,3, \ldots, N \}$ this becomes
\begin{align}
d_t \beta^c(t)
= &- 2 \Im \SCP{q^{\varphi_t}_1 \left( - \Delta_1 + N^{-1/2} \vph(x_1)  \right) \Psi_{N,t}}{ (- \Delta_1) q^{\varphi_t}_1 \Psi_{N,t}}
\nonumber\\
&+ 2 \Im \SCP{q^{\varphi_t}_1 H_1^{eff} \Psi_{N,t}}{(- \Delta_1) q^{\varphi_t}_1 \Psi_{N,t}}
\nonumber\\
&- 2 \Im \SCP{H_1^{eff} q^{\varphi_t}_1  \Psi_{N,t}}{(- \Delta_1) q^{\varphi_t}_1 \Psi_{N,t}}
\nonumber\\
= &- 2 \Im \SCP{q^{\varphi_t}_1 \left( N^{-1/2} \vph(x_1) - \vphc(x_1,t) \right) \Psi_{N,t}}{(- \Delta_1) q^{\varphi_t}_1 \Psi_{N,t}}
\nonumber \\
&- 2 \Im \SCP{\vphc(x_1,t) q^{\varphi_t}_1 \Psi_{N,t}}{ (- \Delta_1) q^{\varphi_t}_1 \Psi_{N,t}}
\nonumber \\
&- 2 \Im \norm{(- \Delta_1) q_1^{\varphi_t} \Psi_{N,t}}^2 
\nonumber\\
= &- 2 \Im \SCP{q^{\varphi_t}_1 \left( N^{-1/2} \vph(x_1) - \vphc(x_1,t) \right) \Psi_{N,t}}{(- \Delta_1) q^{\varphi_t}_1 \Psi_{N,t}}
\nonumber \\
&- 2 \Im \SCP{\vphc(x_1,t) q^{\varphi_t}_1 \Psi_{N,t}}{ (- \Delta_1) q^{\varphi_t}_1 \Psi_{N,t}}.
\end{align}
The identity $q^{\varphi_t}_1 \mathcal{O} = \mathcal{O}p^{\varphi_t}_1 + \mathcal{O} q^{\varphi_t}_1 -  p^{\varphi_t}_1 \mathcal{O}$ (for any operator $\mathcal{O}$) and 
\begin{align}
- \SCP{\vphc(x_1,t) q^{\varphi_t}_1 \Psi_{N,t}}{( - \Delta_1) q^{\varphi_t}_1 \Psi_{N,t}}
&=  \SCP{\left( N^{-1/2} \vph(x_1) - \vphc(x_1,t) \right) q^{\varphi_t}_1 \Psi_{N,t}}{(- \Delta_1) q^{\varphi_t}_1 \Psi_{N,t}}
\nonumber\\
&- \SCP{N^{-1/2} \vph(x_1) q^{\varphi_t}_1 \Psi_{N,t}}{ (- \Delta_1) q^{\varphi_t}_1 \Psi_{N,t}}
\end{align}
lead to
\begin{align}
d_t \beta^c(t) 
\label{eq: Nelson beta-c-1}
&= 2 \Im \SCP{p^{\varphi_t}_1 \left( N^{-1/2} \vph(x_1) - \vphc(x_1,t) \right) \Psi_{N,t}}{(- \Delta_1) q^{\varphi_t}_1 \Psi_{N,t}}
\\ 
\label{eq: Nelson beta-c-2}
&- 2 \Im \SCP{ \left( N^{-1/2} \vph(x_1) - \vphc(x_1,t) \right) p^{\varphi_t}_1 \Psi_{N,t}}{(- \Delta_1) q^{\varphi_t}_1 \Psi_{N,t}}
\\
\label{eq: Nelson beta-c-3}
&- 2 \Im \SCP{ N^{-1/2} \vph(x_1) q^{\varphi_t}_1 \Psi_{N,t}}{ (- \Delta_1) q^{\varphi_t}_1 \Psi_{N,t}}.
\end{align}

\end{proof}

\begin{lemma}
\label{lemma: Nelson time derivative of beta-2}
Let $(\varphi_t,\alpha_t) \in \mathcal{G}_2$ and $\Psi_{N,t}$ be the unique solution of \eqref{eq: Nelson Schroedinger equation microscopic} with initial data $ \Psi_{N,0} \in \left( L_s^2(\mathbb{R}^{3N})  \otimes \mathcal{F} \right) \cap \mathcal{D} \left( \mathcal{N} \right) \cap \mathcal{D} \left(  \mathcal{N}  H_N \right) \cap \mathcal{D}  \left( H_N^2  \right)$ such that $\norm{\Psi_{N,0}} =1$. 
Then, there exists a positive monotone increasing function $C(s)$ of the norms $\norm{\alpha_s}_{L^2(\mathbb{R}^3)}$ and $\norm{\varphi_s}_{H^1(\mathbb{R}^3)}$ such that
\begin{align}
\label{eq: Nelson time derivative of beta-2}
\abs{d_t \beta_2\left( \Psi_{N,t}, \varphi_t, \alpha_t \right)  }
\leq&   \Lambda^4 C(t) \left( \beta_2\left( \Psi_{N,t}, \varphi_t, \alpha_t \right) + N^{-1} \right), 
\nonumber\\
\beta_2\left( \Psi_{N,t}, \varphi_t, \alpha_t \right)
\leq& e^{\Lambda^4 \int_0^t C(s) ds} \left( \beta_2\left( \Psi_{N,0}, \varphi_0, \alpha_0 \right) + N^{-1}  \right)
\end{align}
hold for any $t \in \mathbb{R}_0^+$.
\end{lemma}

\begin{proof}
In order to estimate $d_t \beta^c(t)$ by $\beta$ and $\norm{\nabla_1 q_1^{\varphi_t} \Psi_{N,t}}$ we will integrate by parts and apply Schwarz's inequality. The gradiant will hereby occasionally act on the radiation fields, which will give rise to the vector fields
\begin{align}
( \nabla \vph )(x) = \int d^3k \, \tilde{\eta}(k) k  i 
\left( e^{ikx} a(k) - e^{-ikx} a^*(k)  \right),
\nonumber\\
(\nabla \vphc)(x,t) = \int d^3k \, \tilde{\eta}(k) k i 
\left( e^{ikx} \alpha_t(k) - e^{-ikx} \alpha^*_t(k)  \right).
\end{align}
We define the vector field  $\tilde{\Theta}(k) \coloneqq \tilde{\eta}(k) k$ and its Fourier transform $\Theta$ with
$ \sum_{i=1}^3 ||\Theta^i ||_2^2 \leq \Lambda^4/(16 \pi^2)$. This allows us to obtain the relation
\begin{equation}
( \nabla \vphp )(x) = i \left( \Theta * \Fp  \right)(x),
\quad
( \nabla \vphpc )(x,t) = i \left( \Theta * \Fpc  \right)(x)
\end{equation}
between the positive frequency part of the vector fields and the auxiliary fields~\eqref{eq: Nelson auxilliary fields}.
In analogy to Lemma~\ref{lemma: Nelson field difference versus beta-b} one proves the estimates
\begin{align}
\label{eq: Nelson nabla feld q1 versus beta estimate}
\norm{\left( N^{-1/2} (\nabla \vph )(x_1)  - \left( \nabla \vphc \right)(x_1,t) \right) p_1 \Psi_N}^2
\leq& C \Lambda^4 \left( \beta^b(t) + N^{-1} \right),  
\nonumber \\
\norm{\left( N^{-1/2} (\nabla \vph )(x_1) - \left( \nabla \vphc \right)(x_1,t) \right) q_1 \Psi_N}^2
\leq& C \Lambda^4 \left( \beta^b(t) + N^{-1} \right),  
\nonumber \\
\norm{\left( N^{-1/2} \vph(x_1) - \vphc(x_1,t) \right) \nabla_1 p_1 \Psi_N}^2
\leq& C \Lambda^2 \norm{\nabla \varphi}^2_2 \left( \beta^b(t) + N^{-1} \right).
\end{align}

\noindent
The first term of $d_t \beta^c(t)$ is estimated by
\begin{align}
\abs{\eqref{eq: Nelson beta-c-1}}
&\leq 2  \abs{\SCP{ p^{\varphi_t}_1 \left( N^{-1/2} \vph(x_1) - \vphc(x_1,t) \right) \Psi_{N,t}}{ (- \Delta_1) q^{\varphi_t}_1 \Psi_{N,t}}} 
\nonumber \\
&= 2  \abs{\SCP{\nabla_1 p^{\varphi_t}_1 \left( N^{-1/2} \vph(x_1) - \vphc(x_1,t) \right) \Psi_{N,t}}{ \nabla_1 q^{\varphi_t}_1 \Psi_{N,t}}} 
\nonumber \\
&\leq  \norm{\nabla_1 p_1 \left( N^{-1/2} \vph(x_1) - \vphc(x_1,t) \right) \Psi_N}^2 + \norm{\nabla_1 q_1 \Psi_N}^2 
\nonumber \\
&\leq \norm{\nabla \varphi_t}^2
\norm{\left( N^{-1/2} \vph(x_1) - \vphc(x_1,t) \right) \Psi_N}^2
+ \norm{\nabla_1 q_1 \Psi_N}^2.
\end{align}
Lemma~\ref{lemma: Nelson field difference versus beta-b} gives rise to
\begin{align}
\abs{\eqref{eq: Nelson beta-c-1}}
&\leq C \Lambda^2 \norm{\nabla \varphi_t}^2 \left( \beta^b + N^{-1} \right) + \norm{\nabla_1 q_1 \Psi_N}^2
\nonumber\\
&\leq   \Lambda^2  C(\norm{\varphi_t}_{H^1}) \left( \beta_2(t) + N^{-1} \right).
\end{align}
Likewise, we estimate
\begin{align}
\abs{\eqref{eq: Nelson beta-c-2}}
\leq& 2 \abs{\SCP{ \left( N^{-1/2} \vph(x_1) - \vphc(x_1,t) \right) p^{\varphi_t}_1 \Psi_{N,t}}{ (- \Delta_1) q^{\varphi_t}_1 \Psi_{N,t}}} 
\nonumber \\
=& 2 \abs{\SCP{\nabla_1 \left( N^{-1/2} \vph(x_1)  - \vphc(x_1,t) \right) p_1 \Psi_N}{ \nabla_1 q_1 \Psi_N}} 
\nonumber \\
\leq& \norm{\nabla_1 \left( N^{-1/2} \vph(x_1)  - \vphc(x_1,t) \right) p_1 \Psi_N}^2
+ \norm{\nabla_1 q_1 \Psi_N}^2.
\end{align}
Due to triangular inequality, $\left(a +b \right)^2 \leq 2 \left(a^2 + b^2 \right)$ and \eqref{eq: Nelson nabla feld q1 versus beta estimate}
 this becomes
\begin{align}
\abs{\eqref{eq: Nelson beta-c-2}}
&\leq 2 \norm{\left( N^{-1/2} \vph(x_1)  - \vphc(x_1,t) \right) \nabla_1 p_1 \Psi_N}^2 
\nonumber \\
&+ 2 \norm{\left( N^{-1/2} (\nabla \vph )(x_1) -  ( \nabla \vphc \right)(x_1) ) p_1 \Psi_N}^2 + \norm{\nabla_1 q_1 \Psi_N}^2
\nonumber\\
&\leq  \Lambda^4 C(\norm{\varphi_t}_{H^1}) \left(\beta_2(t) + N^{-1} \right) .
\end{align}
Next, we consider line
\begin{align}
\eqref{eq: Nelson beta-c-3}
= &- 2 \Im \SCP{\nabla_1  N^{-1/2} \vph(x_1) q^{\varphi_t}_1 \Psi_{N,t}}{ \nabla_1 q^{\varphi_t}_1 \Psi_{N,t}}
\nonumber\\
= &- 2 \Im \SCP{  N^{-1/2}  (\nabla \vph)(x_1) q^{\varphi_t}_1 \Psi_{N,t}}{ \nabla_1 q^{\varphi_t}_1 \Psi_{N,t}}
\nonumber\\
& - 2 \Im \SCP{  N^{-1/2} \vph(x_1) \nabla_1 q^{\varphi_t}_1 \Psi_{N,t}}{ \nabla_1 q^{\varphi_t}_1 \Psi_{N,t}}.
\end{align}
The scalar product in the last line is easily shown to be real. This yields
\begin{align}
\eqref{eq: Nelson beta-c-3}
= &- 2 \Im \SCP{  N^{-1/2}  (\nabla \vph)(x_1) q^{\varphi_t}_1 \Psi_{N,t}}{ \nabla_1 q^{\varphi_t}_1 \Psi_{N,t}}
\nonumber\\
= &- 2 \Im \SCP{ \left( N^{-1/2}  (\nabla \vph)(x_1) - (\nabla \vphc)(x_1,t)  \right) q^{\varphi_t}_1 \Psi_{N,t}}{ \nabla_1 q^{\varphi_t}_1 \Psi_{N,t}}
\nonumber\\
&-2  \Im \SCP{   (\nabla \vphc)(x_1,t) q^{\varphi_t}_1 \Psi_{N,t}}{ \nabla_1 q^{\varphi_t}_1 \Psi_{N,t}}.
\end{align}
and allows us to estimate
\begin{align}
\abs{\eqref{eq: Nelson beta-c-3}}
&\leq 2 \abs{\SCP{ \left( N^{-1/2}  (\nabla \vph)(x_1) - (\nabla \vphc)(x_1,t)  \right) q^{\varphi_t}_1 \Psi_{N,t}}{ \nabla_1 q^{\varphi_t}_1 \Psi_{N,t}}}  
\nonumber\\
&+ 2 \abs{\SCP{   (\nabla \vphc)(x_1,t) q^{\varphi_t}_1 \Psi_{N,t}}{ \nabla_1 q^{\varphi_t}_1 \Psi_{N,t}}}
\nonumber \\
&\leq \norm{ \left( N^{-1/2}  (\nabla \vph)(x_1) - (\nabla \vphc)(x_1,t)  \right) q^{\varphi_t}_1 \Psi_{N,t}}^2
+ \norm{ (\nabla \vphc)(x_1,t) q^{\varphi_t}_1 \Psi_{N,t}}^2
\nonumber\\
&+ 2 \norm{\nabla_1 q^{\varphi_t}_1 \Psi_{N,t}}^2
\leq C \Lambda^4 \left( \beta^b(t) + N^{-1} \right) + C \Lambda^4 \norm{\alpha_t}_2^2 \beta^a(t) + 2 \beta^c(t)
\nonumber\\
&\leq  \Lambda^4 C(\norm{\alpha_t}_2)  \left( \beta_2(t) + N^{-1} \right).
\end{align}
Here, we used \eqref{eq: Nelson nabla feld q1 versus beta estimate} and the fact that
\begin{align}
\norm{\left( \nabla \vphc \right)(\cdot,t)}_{\infty} \leq C \Lambda^2 \norm{\alpha_t}_2
\end{align}
holds because of Schwarz's inequality.
In total, we have
\begin{align}
\abs{d_t \beta^c(t)} 
\leq&  \Lambda^4 C( \norm{\varphi_t}_{H^1}, \norm{\alpha_t} ) \left( \beta_2 + N^{-1} \right).
\end{align}
With Lemma~\ref{lemma: Nelson time derivative of beta} this  implies
\begin{align}
\label{eq: groenwall beta-2-a}
\abs{d_t \beta_2\left[ \Psi_{N,t}, \varphi_t, \alpha_t \right]  }
\leq&   \Lambda^4 C(\norm{\varphi_t}_{H^1}, \norm{\alpha_t}) \left( \beta_2\left[ \Psi_{N,t}, \varphi_t, \alpha_t \right] + N^{-1} \right)
\end{align}
Using the shorthand notation $C(t) \coloneqq C(\norm{\varphi_t}_{H^1}, \norm{\alpha_t} )$ we obtain
\begin{align}
\label{eq: groenwall beta-2-b}
\beta_2\left[ \Psi_{N,t}, \varphi_t, \alpha_t \right]
\leq& e^{\Lambda^4 \int_0^t C(s) ds} \left( \beta_2\left[ \Psi_{N,0}, \varphi_0, \alpha_0 \right] + N^{-1}  \right)
\end{align}
by means of Gronwall's lemma. 

\end{proof}

\section{Initial states}
\label{section Nelson initial states}
Subsequently, we are concerned with the initial states of Theorem~\ref{theorem: Nelson main theorem}. 
\begin{lemma}
\label{lemma: Nelson initial states}
Let $\Psi_{N,0} \in \left( L_s^2(\mathbb{R}^{3N}) \otimes \mathcal{F} \right) \cap \mathcal{D} \left( \mathcal{N} \right)$ with $\norm{\Psi_{N,0}} = 1$ and $(\varphi_0,\alpha_0) \in L^2(\mathbb{R}^3) \oplus L^2(\mathbb{R}^3)$ with $\norm{\varphi_0}= 1$. Then
\begin{align}
\beta^a(\Psi_{N,0}, \varphi_0) &\leq  \text{Tr}_{L^2(\mathbb{R}^3)} \abs{\gamma_{N,0}^{(1,0)} - \ket{\varphi_0} \bra{\varphi_0}} = a_N ,
\nonumber \\
\beta^b(\Psi_{N,0},\alpha_0) &= N^{-1} \SCP{ W^{-1}(\sqrt{N \alpha_0}) \Psi_{N,0}}{\mathcal{N} W^{-1}(\sqrt{N \alpha_0}) \Psi_{N,0} }  = b_N .
\end{align}
\end{lemma}

\begin{proof}
The first inequality is a consequence of Lemma~\ref{lemma: Nelson relation between beta and reduced density matrices}. 
Before we prove the second relation we  justify \eqref{eq: Nelson fluctuation dynamics and beta-b}. Therefore, is useful to note that the Weyl operator ($f \in L^2(\mathbb{R}^3)$)
\begin{align}
W(f) = \exp \left( \int d^3k \, f(k) a^*(k) - f^*(k) a(k)  \right)
\end{align}
 is unitary 
 \begin{align}
 W^{-1}(f) = W^*(f) = W(-f)
 \end{align}
 and satisfies\footnote{More information is given for instance in \cite[p.9]{rodnianskischlein}} 
 \begin{equation}
W^*(f) a(k) W(f) = a(k) + f(k) , \quad  W^*(f) a^*(k) W(f) = a^*(k) + f^*(k).
\end{equation}
 This leads to

\begin{align}
\beta^b(\Psi_{N,t},\alpha_t) 
&= \int d^3k \, \norm{\left( N^{-1/2} a(k) - \alpha_t(k) \right) \Psi_{N,t} }^2
\nonumber\\
&= \int d^3k \, \norm{W^*(\sqrt{N} \alpha_t)  \left( N^{-1/2} a(k) - \alpha_t(k) \right)  W(\sqrt{N} \alpha_t) W^*(\sqrt{N} \alpha_t)   \Psi_{N,t} }^2
\nonumber\\
&= \int d^3k \, \norm{N^{-1/2} a(k) W^*(\sqrt{N} \alpha_t ) \Psi_{N,t}}^2
\nonumber\\
&= N^{-1} \SCP{W^*(\sqrt{N} \alpha_t) e^{-i H_N t} \Psi_{N,0}}{ \mathcal{N} W^*(\sqrt{N} \alpha_t) e^{-i H_N t} \Psi_{N,0} }.
\end{align}
Let 
\begin{align}
\mathcal{U}_N(t;0) \coloneqq W^*(\sqrt{N} \alpha_t) e^{- i H_N t} W(\sqrt{N} \alpha_0)
\end{align}
denote the fluctuation dynamics then
\begin{align}
\label{eq: Nelson relation between beta-b and fluctuation dynamics}
\beta^b(\Psi_{N,t}, \alpha_t)
&= N^{-1} \SCP{\mathcal{U}_N(t;0) W^{-1}(\sqrt{N \alpha_0}) \Psi_{N,0}}{\mathcal{N} \mathcal{U}_N(t;0) W^{-1}(\sqrt{N \alpha_0}) \Psi_{N,0} }
\end{align}
follows from the unitarity of the Weyl operator. In particular, we have
\begin{align}
\label{eq: Nelson relation between beta-b and fluctuation dynamics initial state}
\beta^b(\Psi_{N,0},\alpha_0) = N^{-1} \SCP{ W^{-1}(\sqrt{N \alpha_0}) \Psi_{N,0}}{\mathcal{N} W^{-1}(\sqrt{N \alpha_0}) \Psi_{N,0} }  = b_N .
\end{align}

\end{proof}

\noindent
In the following, we consider initial states of product form~\eqref{eq: Nelson initial product state}.
\begin{lemma}
\label{lemma: Nelson coherent states as initial states}
Let $(\varphi_0,\alpha_0) \in H^2(\mathbb{R}^3) \oplus L_1^2(\mathbb{R}^3)$ with $\norm{\varphi_0}=1$ and $\Psi_{N,0} = \varphi_0^{\otimes N} \otimes W(\sqrt{N} \alpha_0) \Omega$. Then
\begin{align}
\label{eq: Nelson coherent states as initial states 1}
a_N &= \text{Tr}_{L^2(\mathbb{R}^3)} \abs{\gamma_{N,0}^{(1,0)} - \ket{\varphi_0} \bra{\varphi_0}} = 0 ,
\\
\label{eq: Nelson coherent states as initial states 2}
b_N &=  N^{-1} \SCP{ W^{-1}(\sqrt{N \alpha_0}) \Psi_{N,0}}{\mathcal{N} W^{-1}(\sqrt{N \alpha_0}) \Psi_{N,0} }  = 0
\; \text{and}
\\
\label{eq: Nelson coherent states as initial states 3}
\Psi_{N,0} &\in \left( L_s^2(\mathbb{R}^{3N})  \otimes \mathcal{F} \right) \cap \mathcal{D} \left( \mathcal{N} \right) \cap  \mathcal{D} \left(  \mathcal{N}  H_N \right) .
\end{align} 
Let $(\varphi_0,\alpha_0) \in H^4(\mathbb{R}^3) \oplus L_2^2(\mathbb{R}^3)$ with $\norm{\varphi_0}=1$
then 
\begin{align}
\label{eq: Nelson coherent states as initial states 5}
c_N &= \norm{\nabla_1 q_1^{\varphi_0} \Psi_{N,0}}^2 = 0
\\
\label{eq: Nelson coherent states as initial states 4}
\Psi_{N,0} &\in \left( L_s^2(\mathbb{R}^{3N})  \otimes \mathcal{F} \right) \cap \mathcal{D} \left( \mathcal{N} \right) \cap \mathcal{D} \left(  \mathcal{N}  H_N \right) \cap \mathcal{D}  \left( H_N^2  \right).
\end{align}
\end{lemma}

\begin{proof}
From the definition of the one-particle reduced density matrix and \eqref{eq: Nelson relation between beta-b and fluctuation dynamics initial state} we directly obtain the relations~\eqref{eq: Nelson coherent states as initial states 1} and~\eqref{eq: Nelson coherent states as initial states 2}.  Equation~\eqref{eq: Nelson coherent states as initial states 5} holds because $\Psi_{N,0}$ is in the kernel of the projector $q_1^{\varphi_0}$.
In order to show~\eqref{eq: Nelson coherent states as initial states 3} we point out that
\begin{align}
\Psi_{N,0}^{(n)}(X_N,K_n) = \prod_{i=1}^N \varphi_0(x_i) e^{- N \norm{\alpha_0}^2/2} (n!)^{-1/2} \prod_{j=1}^n (N)^{1/2} \alpha_0(k_j)
\end{align} 
follows from the definition of the the Weyl operators~\cite[p.8]{rodnianskischlein}.
A direct calculation gives
\begin{align}
\sum_{n=1}^{\infty} n^2 \norm{\Psi_{N,0}^{(n)}}^2 = N \norm{\alpha_0}^2 + N^2 \norm{\alpha_0}^4.
\end{align}
Hence, $\Psi_{N,0}^{(n)} \in \mathcal{D}(\mathcal{N})$ (see \eqref{eq: Nelson number operator domain}).
Moreover, we have $\Psi_{N,0} \in \mathcal{D}(\sum_{i=1}^N - \Delta_i )$ because $\varphi_0 \in H^2(\mathbb{R}^3)$.
A straightforward estimate leads to 
\begin{align}
\sum_{n=1}^{\infty} \int d^{3N}x \, d^{3n}k \, \abs{\sum_{j=1}^n w(k_j)}^2  \abs{\Psi_{N,0}^{(n)}(X_N, K_n)}^2
\leq C(N, \norm{\alpha_0}_{L_1^2(\mathbb{R}^3)}).
\end{align}
From~\eqref{eq: Nelson field energie domain} we then conclude $\Psi_{N,0}^{(n)} \in \mathcal{D}(H_f)$ and $\Psi_{N,0}^{(n)} \in \mathcal{D}(H_N) = \mathcal{D}(\sum_{i=1}^N - \Delta_i ) \cap \mathcal{D}(H_f)$.
Similarly, one derives
\begin{align}
\sum_{n=1}^N n^2 \norm{(H_N \Psi_{N,0})^{(n)}}^2
&\leq C \sum_{n=1}^{\infty} n^2 \Big(  \Big| \Big| \sum_{j=1}^N \Delta_j \Psi_{N,0}^{(n)} \Big| \Big|^2
+ \Big| \Big| \sum_{j=1}^N N^{-1/2} ( \vph(x_j) \Psi_{N,0})^{(n)} \Big| \Big|^2  \Big)
\nonumber\\
&+ C \sum_{n=1}^{\infty} n^2 \norm{(H_f \Psi_{N,0})^{(n)}}^2
\leq C(N, \Lambda, \norm{\varphi_0}_{H^2(\mathbb{R}^3)}, \norm{\alpha_0}_{L_1^2(\mathbb{R}^3)}).
\end{align}
and concludes $\Psi_{N,0} \in \mathcal{D}(\mathcal{N} H_N) =  \big\{ \Psi_N \in \mathcal{D}(H_N) : H_N \Psi_N \in \mathcal{D}(\mathcal{N})  \big\}$.
In order to show~\eqref{eq: Nelson coherent states as initial states 4} we would like to note that $(\varphi_0, \alpha_0 ) \in (H^4(\mathbb{R}^3), L_2^2(\mathbb{R}^3))$, $\abs{\cdot}^2 \tilde{\eta} \in L^2(\mathbb{R}^3)$ and $\tilde{\eta} \in L^2(\mathbb{R}^3)$ imply $H_N \Psi_{N,0} \in \mathcal{D} (\sum_{i=1}^N - \Delta_i)$.
By means of the estimate
\begin{align}
\sum_{n=1}^{\infty} d^{3N}x \, d^{3n}k \,  \abs{\sum_{j=1}^n w(k_j)}^2  \abs{(H_N \Psi_{N,0})^{(n)}(X_N, K_n)}^2
&\leq C(N, \Lambda, \norm{\varphi_0}_{H^2(\mathbb{R}^3)}, \norm{\alpha_0}_{L_2^2(\mathbb{R}^3)})
\end{align}
one obtains $H_N \Psi_{N,0} \in \mathcal{D}(H_f)$. In total, we have
$H_N \Psi_{N,0} \in \mathcal{D}(H_N)$ and $\Psi_{N,0} \in \mathcal{D}(H_N^2)$.
\end{proof}

\section{Proof of Theorem~\ref{theorem: Nelson main theorem}}
In order to finish the proof of Theorem~\ref{theorem: Nelson main theorem} we remark that
Lemma~\ref{lemma: Nelson initial states} leads to
\begin{align}
\beta(\Psi_{N,0}, \varphi_0, \alpha_0) &\leq a_N + b_N ,
\nonumber\\
\beta_2 (\Psi_{N,0}, \varphi_0, \alpha_0) &\leq a_N + b_N + c_N .
\end{align}
We then choose for a given time $t \in \mathbb{R}_0^+$ the number $N$ of charged particles large enough such that the values of $\beta(\Psi_{N,t}, \varphi_t, \alpha_t)$ in \eqref{eq: Nelson time derivative of beta} and $\beta_2(\Psi_{N,t}, \varphi_t, \alpha_t)$ in \eqref{eq: Nelson time derivative of beta-2} are smaller than one and derive Theorem~\ref{theorem: Nelson main theorem} by means of Lemma~\ref{lemma: Nelson relation between beta and reduced density matrices}.

\section*{Acknowledgments}

We would like to thank Dirk Andr\'{e} Deckert, Marco Falconi and David Mitrouskas  for helpful discussions.
N.L. gratefully acknowledges financial support by the Cusanuswerk and the European Research Council (ERC) under the European Union's Horizon 2020 research and innovation programme (grant agreement No 694227).
The article appeard in slightly different form in one of the author's (N.L.) Ph.D. thesis \cite{leopold2}.

{}

\end{document}